\documentclass[10pt]{elsarticle}

\usepackage{amsthm,amsmath,amsfonts,amssymb,amscd,mathrsfs}
\usepackage{txfonts}
\usepackage{supertabular,soul}
\usepackage[usenames,dvipsnames]{xcolor}
\usepackage{tikz, graphicx,color,geometry}
 \usepackage{multirow}
 \usetikzlibrary{arrows}
\usepackage[pdftex,
            pdfauthor={Barrett},
            pdftitle={Automorphisms, Equitable Partitions, and Spectral Graph Theory},
            pdfsubject={Equitable Partitions, Spectral Graph Theory},
            pdfkeywords={Equitable Partitions, Spectral Graph Theory}]{hyperref}

\usepackage{bbm} 
\usepackage{hyperref}
\usepackage{yfonts}
\usepackage{eucal}
\usepackage{overpic}
\usetikzlibrary{calc}
\usepackage{enumitem}

\newcommand{\comment}[1]{}

\makeatletter
\def\ps@pprintTitle{%
\let\@oddhead\@empty
\let\@evenhead\@empty
\let\@oddfoot\@empty
\let\@evenfoot\@oddfoot}
\makeatother

\def\big{\bigskip}

\newtheorem{theorem}{Theorem}
\newtheorem{definition}{Definition}
\newtheorem{thm}{Theorem}[section]
\newtheorem{corollary}[thm]{Corollary}

\newtheorem{lemma}{Lemma}

\newtheorem{prop}[thm]{Proposition}

\newtheorem{example}[thm]{Example}

\theoremstyle{remark}

\providecommand*{\propertyautorefname}{Property}

\let\oldmarginpar\marginpar
\renewcommand\marginpar[1]{\oldmarginpar[\raggedleft\footnotesize #1]%
{\raggedright\footnotesize #1}}

\begin{document}
\begin{frontmatter}


\title{Spectral and Dynamic Consequences of Network Specialization}

\author[leo]{Leonid Bunimovich}
\address[leo]{School of Mathematics, Georgia Institute of Technology, 686 Cherry Street, Atlanta, GA 30332, USA
bunimovich@math.gatech.edu}
\author[dj]{DJ Passey}
\address[dj]{Department of Mathematics, Brigham Young University, Provo, UT 84602, USA, djpasseyjr@gmail.com}
\author[dallas]{Dallas Smith}
\address[dallas]{Department of Mathematics, Brigham Young University, Provo, UT 84602, USA, dallas.smith@mathematics.byu.edu}
\author[ben]{Benjamin Webb}
\address[ben]{Department of Mathematics, Brigham Young University, Provo, UT 84602, USA, bwebb@mathematics.byu.edu}

\begin{abstract}
One of the hallmarks of real networks is their ability to perform increasingly complex tasks as their topology evolves. To explain this, it has been observed that as a network grows certain subsets of the network begin to specialize the function(s) they perform. A recent model of network growth based on this notion of specialization has been able to reproduce some of the most well-known topological features found in real-world networks including right-skewed degree distributions, the small world property, modular as well as hierarchical topology, etc. Here we describe how specialization under this model also effects the spectral properties of a network. This allows us to give conditions under which a network is able to maintain its dynamics as its topology evolves. Specifically, we show that if a network is intrinsically stable, which is a stronger version of the standard notion of global stability, then the network maintains this type of dynamics as the network evolves. This is one of the first steps toward unifying the rigorous study of the two types of dynamics exhibited by networks. These are the \emph{dynamics of} a network, which is the study of the topological evolution of the network's structure, modeled here by the process of network specialization, and the \emph{dynamics on} a network, which is the changing state of the network elements, where the type of dynamics we consider is global stability. The main examples we apply our results to are recurrent neural networks, which are the basis of certain types of machine learning algorithms.
\end{abstract}


\end{frontmatter}

\section{Introduction}
Networks studied in the biological, social and technological sciences perform various tasks. How well these tasks are carried out depends both on the network's \emph{topology}, i.e. the network's structure of interactions, as well as the dynamics of the network elements. For instance, in the biological setting neuronal networks are responsible for complicated processes related to cognition and memory, which are based on the network's structure of connections as well as the electrical dynamics of the network neurons \cite{BS09}. The function and performance of social networks such as Facebook and Twitter depend on the social interactions between individuals and on the local/global structure of established relationships. The performance of technological networks such as the internet is based both on the topology of the network's links, e.g. connections between routers, and the router's ability to manage requests.

A major goal of network science is to understand how a network's \emph{topology} and the dynamics of the network's elements effects the network's ability to carry out its function. This goal is complicated by the fact that not only are the network's elements \emph{dynamic}, i.e. have a state that changes over time, but the network's topology also evolves over time. That is, most real-world networks are dynamic in two distinct ways.

The changing state of the network elements is referred to as the \emph{dynamics on} the network. For instance, the dynamics on  the World Wide Web consists of internet traffic, or how many individuals navigate between different web page over time. In a neural network the electrical behavior of the individual neurons is the dynamics on the network. In a distribution network of warehouses and delivery routes, the way in which inventory levels at each warehouse change over time describes the \emph{dynamics on} the network.

In contrast, the structural evolution of a network's topology is referred as the \emph{dynamics of} the network. In a transportation network the topology of the network evolves as new roads are built to increase the number of routes between destinations. In the World Wide Web new links between web pages are added for similar reasons as are new warehouses and delivery routes in distribution networks. In other networks such as social networks the network topology evolves as individuals form new relations, etc.

In most real networks there is an interplay between the \emph{dynamics on} the network and the \emph{dynamics of} the network. For instance, in transportation, technological, and information networks new routes are added in response to areas of high traffic, e.g. the addition of new road in transportation systems, new routers and connections to the internet, new web pages to the World Wide Web, etc. These new routes in turn create new traffic patterns that lead to different areas of high traffic and the creation of new routes.

Determining how these two types of dynamics influence each other is an especially difficult problem. In fact, most studies consider either \emph{dynamics on} the network or \emph{dynamics of} the network but not the interplay of the two \cite{Newman2006}. There are some exceptions. For instance, the area of adaptive networks uses the dynamics on the network to change or ``adapt" the network's topology \cite{GS09}. However, to the best of the authors knowledge, results in this area are strictly numerical and little if anything is rigourously known. A major goal of this paper is to take steps towards unifying the rigorous analysis of the \emph{dynamics on} and \emph{of} networks.

A complication to achieving this objective is that there are many network growth models describing the dynamics of a network and many models describing the dynamics of a network. The most well-known of the network growth models is the Barabasi and Albert model \cite{Bara99} and its predecessor the Price model \cite{Price76}, in which elements are added one by one to a network and are preferentially attached to vertices with high degree or some variant of this process
\cite{Bara00,Doro00,Krap01}. Other standard models include \emph{vertex copying models} in which a new element is added to a network by choosing another network element uniformly at random and connecting the new element to this element's neighbors \cite{Klein99,Sole02,Vaz03}. The third type of network growth models are \emph{network optimization models} where the topology of the network is evolved to minimize some global constraint, e.g. operating costs vs. travel times in a transportation network \cite{Ferrer03,Gastner06}.

These models are devised to create networks that exhibit some of the most widely observed features found in real networks. This includes right-skewed degree distributions, high \emph{clustering coefficients}, the \emph{small-world property}, etc. (see \cite{Newman10} for more details on these properties). Although these model are successful in a number of ways, currently, little is known regarding how the \emph{dynamics on} a network is effected by the growth described in these models. Similarly, relatively little is known regarding how these models effect the spectral properties of networks. This is of interest as a network's eigenvalues and eigenvectors are used in defining network centralities, community detection algorithms, determining network dynamics, etc. \cite{Newman10}.

One of the hallmarks of real networks is that as their topology evolves so does their ability to perform increasingly complex versions of their function. This happens, for instance, in neural networks, which become more modular in structure as individual parts of the brain become increasingly specialized in function \cite{Sporns13}. Similarly, gene regulatory networks can specialize the activity of existing genes to create new patterns of gene behavior \cite{Esp10}. In technological networks such as the internet, this differentiation of function is also observed and is driven by the need to handle and more efficiently process an increasing amount of information.

Previously, in \cite{BSW18} the authors introduced the \emph{specialization model} of network growth built on this notion of specialization. This model is based on the idea that as a network specializes the function of one or more of its \emph{components}, i.e. some subnetwork that performs a specific function, a number of copies of this component are created that are linked to the network in a way similar to the original component. These new copies ``specialize" the function of the original component in that they carry out only those functions requiring these links (cf. Figure \ref{Fig:1} and \ref{Fig:2}). Repeated application of this specialization process results in networks with a highly modular \cite{Milo02,Newman2006.1}, hierarchical \cite{Clauset08,Leskovec2008}, as well as sparse topology \cite{N03,HG08} each of which is a defining feature of real networks, which are not captured by previous models (see \cite{BSW18} for more details).

In contrast to studying the topological features produced by this model, in this paper we consider how specialization effects the spectral properties of a network. This is done by using, and in particular extending, the theory of isospectral network transformations \cite{BW12,BWBook}, which describes how certain changes in a network's structure effect the network's \emph{spectrum}, i.e. the eigenvalues associated with the network's weighted or unweighted adjacency matrix. We show that if a network is specialized then the spectrum of the resulting network is that of the original network together with the eigenvalues of the specialized components (see Section \ref{sec3}, Theorem \ref{thm1}). Additionally, using the theory of isospectral transformations we show how the eigenvectors of a network are effected by specialization and specifically its effect on eigenvector centrality (see Section, \ref{sec3} Theorem \ref{prop:0})

As a network's dynamics can be related to the network's spectrum we can in certain cases determine how specialization of a network will effect the \emph{dynamics on} the network. The type of dynamics we consider here is global stability, which is an important property in a number of systems including neural networks \cite{Cao2003,Cheng2006,SChena2009,MCohen1983,LTao2011}, network epidemic models \cite{Wang2008}, and in the study of congestion on computer networks \cite{Alpcan2005}. The main example(s) we consider here and apply our results to are recurrent neural networks, which model the electrical activity in the brain and which form the basis of certain algorithms in machine learning \cite{Sch2015}.

In a \emph{stable} dynamical network the network always evolves to the same unique state irrespective of the network's initial state. We show that if a dynamical network is \emph{intrinsically stable}, which is a stronger form of a standard notion of stability (see Definition \ref{def:intrinsic} or \cite{BW13} for more details), then any specialized version of this network will also be intrinsically stable (see Section \ref{sec4}, Theorem \ref{thm:evostability}). Hence, network growth via specialization will not destabilize the network's dynamics if the network has this stronger version of stability. This is potentially useful in real-world applications since network growth can have a destabilizing effect on a network. For instance, a well-known example of this phenomena is cancer, which is the abnormal growth of cells that can impair the function and ultimately lead to the failure of the system.

Although networks exhibit many other types of dynamics including multistability, periodicity, and synchronization the reason we study stability is because of its relative simplicity and use it as a first step towards rigourously understanding the interplay of network dynamics and network growth. Importantly, our results suggest that if a network's growth is due to specialization and the network has a ``strong" form of dynamics, e.g. intrinsic stability, the network can maintain its dynamics as its structure evolves. It is worth noting that the ability to maintain dynamics during periods of growth is an important feature found in real networks, e.g. the cellular network of a beating heart maintain's its periodic beating as an adolescent heart grows into an adult heart.

This paper is organized as follows. In Section \ref{sec2} we describe the specialization model of network growth. In Section \ref{sec3} we give our main results regarding the spectral properties of specialized networks, which describe the effect of a specialization on the eigenvalues and eigenvectors associated with a network or graph. In Section \ref{sec4} we adapt the method of graph specializations to evolve the structure of a general class of dynamical systems used to model network dynamics, which include recurrent neural networks. We show that if such dynamical systems, which we refer to as dynamical networks, are intrinsically stable then any specialized version of the network is also intrinsically stable. Section \ref{sec5} discusses a number of variants of the specialization method introduced in \cite{BSW18} and described in Section \ref{sec2} as well as their spectral and dynamic consequences to a network. Proofs of the paper's main results together with the necessary parts of the theory of isospectral network transformations are placed in Section \ref{appendix}. The last section, Section \ref{conc} contains some concluding remarks.

\section{The Specialization Process}\label{sec2}

The \emph{topology} of a network, which is the network's structure of interactions, is most often represented by a graph. A \emph{graph} $G=(V,E,\omega)$ is composed of a \emph{vertex set} $V$, an \emph{edge set} $E$, and a function $\omega$ used to weight the edges of the graph. The vertex set $V$ represents the \emph{elements} of the network, while the edges $E$ represent the links or \emph{interactions} between these network elements. The weights of the edges, given by $\omega$, typically give some measure of the \emph{strength} of these interactions. Here we consider weights that are real numbers, which account for the vast majority of weights used in network analysis \cite{Newman10}.

For the graph $G=(V,E,\omega)$ we let $V=\{v_1,\dots,v_n\}$, where $v_i$ represents the $i$th network element. An edge between vertices $v_i$ and $v_j$ can be either directed or undirected depending on the particular network. In the undirected case in which an interaction is reciprocal, e.g. a \emph{friendship} in a social network, one can consider an undirected edge to be two directed edges: one edge pointing from the first to the second element, the other pointing from the second to the first element. We can, therefore, consider any graph to be a directed graph. For a graph $G=(V,E,\omega)$ we let $e_{ij}$ denote the directed edge that begins at $v_j$ and ends at $v_i$. In terms of the network, the edge $e_{ij}$ belongs to the edge set $E$ if the $j$th network element has some direct influence or is linked to the $i$th network element in some way.

The edges $E$ of a graph can also either be \emph{weighted} or \emph{unweighted}. If the edges of a graph $G$ are unweighted we write $G=(V,E)$. However, any unweighted graph can be considered weighted by giving each edge unit weight. The class of graphs we consider, without any loss in generality, in this paper are weighted directed graphs. To motivate the specialization model introduced in \cite{BSW18} of a graph $G$, representing some network, we consider the following example of disambiguation in a Wikipedia network, a subset of the World Wide Web.

\begin{figure}
  \begin{overpic}[scale=.2]{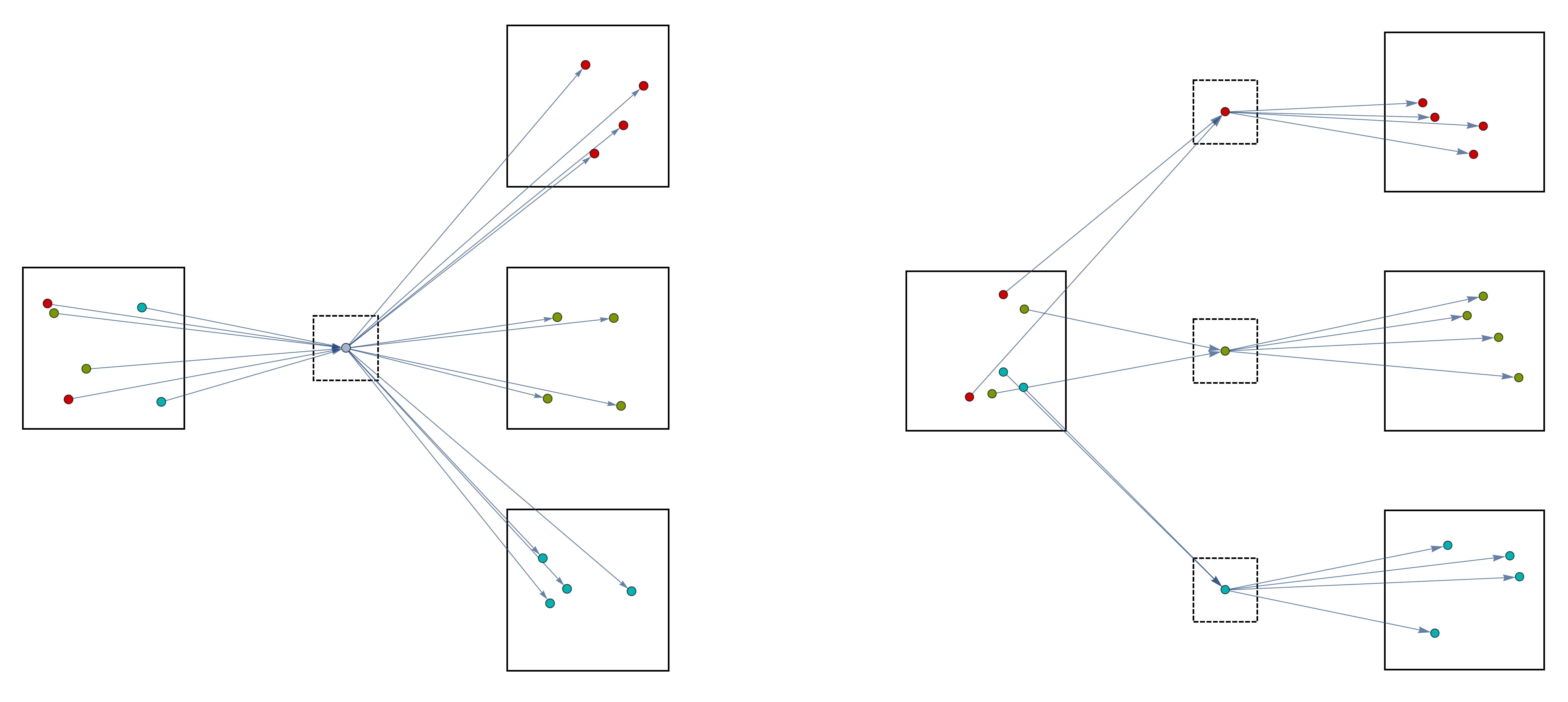}

    \put(10,-1.5){Undifferentiated Mercury Page}

    \put(18.5,28){\small{\emph{Mercury}}}
    \put(20,26){\small{\emph{Page}}}
    \put(43,38){\small{\emph{Mythology}}}
    \put(43,36){\small{\emph{Pages}}}
    \put(43,23){\small{\emph{Planet}}}
    \put(43,21){\small{\emph{Pages}}}
    \put(43,7){\small{\emph{Element}}}
    \put(43,5){\small{\emph{Pages}}}

    \put(65,-1.5){Disambiguated Mercury Pages}

    \put(74.5,42.5){\small{\emph{Mercury}}}
    \put(73,40.5){\small{\emph{(Mythology)}}}
    \put(74.5,28){\small{\emph{Mercury}}}
    \put(73.75,26){\small{\emph{(Element)}}}
    \put(74.5,12.5){\small{\emph{Mercury}}}
    \put(74.75,10.5){\small{\emph{(Planet)}}}

    \put(99,38){\small{\emph{Mythology}}}
    \put(99,36){\small{\emph{Pages}}}
    \put(99,23){\small{\emph{Planet}}}
    \put(99,21){\small{\emph{Pages}}}
    \put(99,7){\small{\emph{Element}}}
    \put(99,5){\small{\emph{Pages}}}
    \end{overpic}
    \vspace{0.25cm}
  \caption{Disambiguation of the Wikipedia page on ``Mercury" into three distinct webpages, which are respectively Mercury the mythological figure, Mercury the planet, and mercury the element. The colors red, green, and blue represent web pages that are respectively mythology, planets, and elements pages.}\label{Fig:1}
\end{figure}
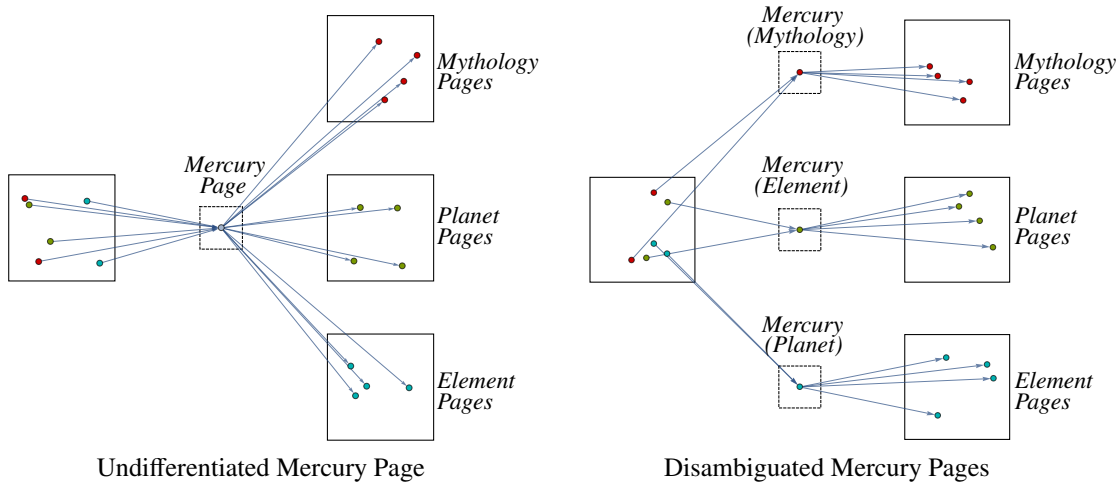

\begin{example}
\textbf{(Wikipedia Disambiguation)} The website Wikipedia is a collection of webpages consisting of articles that are linked by topic. The website evolves as new articles are either added, linked, and/or modified within the existing website. One of the ways articles are modified is via the process of disambiguation. That is, if an article's content is deemed to refer to a number of distinct topics then the article can be \emph{disambiguated} by separating the article into a number of new articles, each on a more specific or specialized topic than the original.

Wikipedia's own page on disambiguation gives the example that the word ``Mercury" can refer to either Mercury the mythological figure, Mercury the planet, or mercury the element. To emphasize these differences, the Wikipedia page on Mercury has since been disambiguated into three pages; one for each of these subcategories. Users arriving at the Wikipedia ``Mercury" page from some related web page are redirected to these pages. The result of this disambiguation is shown in Figure \ref{Fig:1}. In the original undifferentiated Mercury page users arriving from other pages could presumably find links to other mythology, planet and element pages (see Figure \ref{Fig:1}, left). After the page was disambiguated users were linked to only those relevant to the particular ``Mercury", e.g. mythology, planet, or element (see Figure \ref{Fig:1}, right). In terms of the topology of the network, this disambiguation results in the creation of a number of new ``Mercury" pages each of which is linked to a subset of pages that were linked to the original Mercury page.
\end{example}

In this example, growth via disambiguation is a result of the number of new copies of the original webpage. However, what is important to the functionality of the new specialized network is that the way in which these new copies are linked to the unaltered pages reflects the topology of the original network. In the specialization model the way in which we link these new components is by separating out the paths and cycles on which these components lie, in a way that mimics the original network structure.

To describe the model of network specialization introduced in \cite{BSW18} and its spectral and dynamic consequences we first need to describe the paths and cycles of a graph. A \emph{path}\index{path} $P$ in the graph $G=(V,E,\omega)$ is an ordered sequence of distinct vertices $P=v_1,\dots,v_m$ in $V$ such that $e_{i+1,i}\in E$ for $i=1,\dots,m-1$. If the first and last vertices $v_1$ and $v_m$ are the same then $P$ is a \emph{cycle}\index{cycle}. If it is the case that a cycle contains a single vertex then we call this cycle a \emph{loop}\index{loop}.

Another fundamental concept that we require is the notion of a strongly connected component. A graph $G=(V,E,\omega)$ is \emph{strongly connected} if for any pair of vertices $v_i,v_j\in V$ there is a path from $v_i$ to $v_j$ or, in the trivial case, $G$ consists of a single vertex. A \emph{strongly connected component} of a graph $G$ is a subgraph that is strongly connected and is maximal with respect to this property.

Because we are concerned with evolving the topology of a network in ways that preserve, at least locally, the network's topology we will also need the notion of a graph restriction. For a graph $G=(V,E,\omega)$ and a subset $B\subseteq V$ we let $G|_{B}$ denote the \emph{restriction} of the graph $G$ to the vertex set $B$, which is the subgraph of $G$ on the vertex set $B$ along with any edges of the graph $G$ between the vertices in $B$. We let $\bar{B}$ denote the \emph{complement} of $B$, so that the restriction $G|_{\bar{B}}$ is the graph restricted to the complement of vertices in $B$.

The key to specializing the structure of a graph is to look at the strongly connected components of the restricted graph $G|_{\bar{B}}$. If $C_1,\dots,C_m$ denote these strongly connected components then we need the collection of paths or cycles of these components, which we refer to as \emph{components branches}.

\begin{definition}\label{def:componentbranch} \textbf{(Component Branches)}
For a graph $G=(V,E,\omega)$ and vertex set $B\subseteq V$ let $C_1,\dots,C_m$ be the strongly connected components of $G|_{\bar{B}}$. If there are edges $e_0,e_1,\dots,e_m\in E$ and vertices $v_i,v_j\in B$ such that\\
\indent (i) $e_k$ is an edge from a vertex in $C_k$ to a vertex in $C_{k+1}$ for $k=1,\dots,m-1$;\\
\indent (ii) $e_0$ is an edge from $v_i$ to a vertex in $C_1$; and\\
\indent (iii) $e_m$ is an edge from a vertex in $C_m$ to $v_j$, then we call the ordered set
\[
\beta=\{v_i,e_{0},C_1,e_{1},C_2,\dots,C_m,e_{m},v_{j}\}
\]
a \emph{path of components} in $G$ with respect to $B$. If $v_i=v_j$ then $\beta$ is a \emph{cycle of components}. We call the collection $\mathcal{B}_B(G)$ of these paths and cycles the \emph{component branches} of $G$ with respect to the base set of vertices $B$.
\end{definition}

The sequence of components $C_1,\dots,C_m$ in this definition can be empty in which case $m=0$ and $\beta$ is the trivial path $\beta=\{v_i,v_j\}$ or loop if $v_i=v_j$. It is worth emphasizing that each branch $\beta\in\mathcal{B}_B(G)$ is a subgraph of $G$. Consequently, the edges of $\beta$ inherit the weights they had in $G$ if $G$ is weighted. If $G$ is unweighted then its component branches are likewise unweighted.

Once a graph has been decomposed into its various branches we construct the specialized version of the graph by merging these branches as follows.

\begin{definition} \textbf{(Graph Specialization)}\label{def:exp}
Suppose $G=(V,E,\omega)$ and $B\subseteq V$. Let $\mathcal{S}_B(G)$ be the graph which consists of the component branches $\mathcal{B}_{B}(G)=\{\beta_1,\dots,\beta_{\ell}\}$ in which we \emph{merge}, i.e. identify, each vertex $v_i\in B$ in any branch $\beta_j$ with the same vertex $v_i$ in any other branch $\beta_k$. We refer to the graph  $\mathcal{S}_B(G)$ as the \emph{specialization} of $G$ over the \emph{base} vertex set $B$.
\end{definition}

A specialization of a graph $G$ over a base vertex set $B$ is a two step process. The first step is the construction of the graph's component branches. The second step is the merging of these components into a single graph. We note that, in a component branch $\beta\in\mathcal{B}_B(G)$ only the first and last vertices of $\beta$ belong to the base $B$. The specialized graph $\mathcal{S}_B(G)$ is therefore the collection of branches $\mathcal{B}_B(G)$ in which we identify an endpoint of two branches if they are the same vertex. This is demonstrated in the following example.

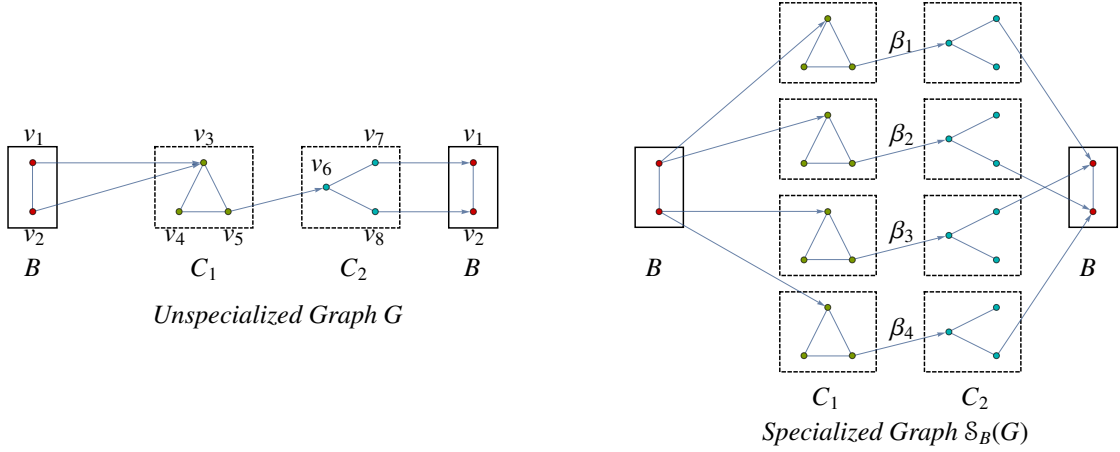
\begin{figure}
  \begin{overpic}[scale=.28]{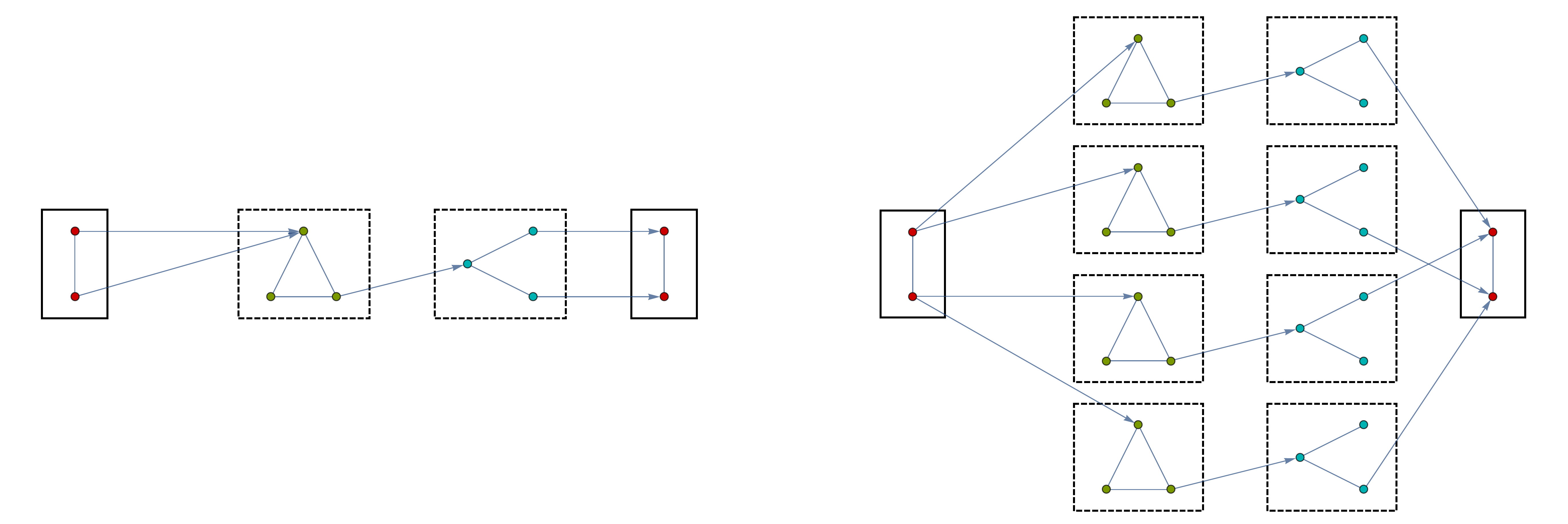}

    \put(15,5){\emph{Unspecialized Graph} $G$}

    \put(4,9){$B$}
    \put(18.25,9){$C_1$}
    \put(31,9){$C_2$}
    \put(41.5,9){$B$}

    \put(4,20.5){$v_1$}
    \put(4,12){$v_2$}
    \put(18.5,20.5){$v_3$}
    \put(16,12){$v_4$}
    \put(21,12){$v_5$}
    \put(28.5,18){$v_6$}
    \put(33,20.5){$v_7$}
    \put(33,12){$v_8$}
    \put(41.5,20.5){$v_1$}
    \put(41.5,12){$v_2$}

    \put(67,-5){\emph{Specialized Graph} $\mathcal{S}_B(G)$}

    \put(57,9){$B$}
    \put(71.25,-2){$C_1$}
    \put(84,-2){$C_2$}
    \put(94,9){$B$}

    \put(78,28.5){$\beta_1$}
    \put(78,20.5){$\beta_2$}
    \put(78,12.25){$\beta_3$}
    \put(78,4){$\beta_4$}
    \end{overpic}
    \vspace{0.5cm}
  \caption{The unweighted graph $G=(V,E)$ is shown left. The graph has components $C_1$ and $C_2$ with respect to the vertex set $B=\{v_1, v_2\}$. These components are indicated by the dashed boxes, which are the strongly connected components of the restricted graph $G|_B$. The specialization $\mathcal{S}_B(G)$ of $G$ over the base $B$ is shown right with nontrivial branches $\beta_1,\dots,\beta_4$.}\label{Fig:2}
\end{figure}

\begin{example}\label{ex:2} \textbf{(Constructing Graph Specializations)}
Consider the \emph{unweighted} graph $G=(V,E)$ shown in Figure \ref{Fig:2} (left). For the base vertex set $B=\{v_1,v_2\}$, which are shown in red in the figure once on the left and once on the right, the specialization $\mathcal{S}_B(G)$ is constructed as follows.\\

\noindent\emph{Step 1:} \emph{Construct the branch components of $G$ with respect to $B$.} The graph $G|_{\bar{B}}$ has the strongly connected components $C_1=G|_{\{v_3,v_4,v_5\}}$ and $C_2=G|_{\{v_6,v_7,v_8\}}$, which are the green and blue vertices in Figure \ref{Fig:2} (right), respectively. The set $\mathcal{B}_B(G)$ of all paths and cycles of components beginning and ending at vertices in $B$ consists of the component branches
\begin{align*}
\beta_1&=\{v_1,e_{31},C_1,e_{65},C_2,e_{17},v_1\} \ \ \ \ \ \ \ \ \ \ \ \ \ \beta_2=\{v_1,e_{31},C_1,e_{65},C_2,e_{28},v_2\}\\
\beta_3&=\{v_2,e_{32},C_1,e_{65},C_2,e_{17},v_1\} \ \ \ \ \ \ \ \ \ \ \ \ \ \beta_4=\{v_2,e_{32},C_1,e_{65},C_2,e_{28},v_2\};
\end{align*}
which can be seen in Figure \ref{Fig:2} (right) along with the trivial branches $\{v_1,v_2\}$ and ${v_2,v_1}$.\\

\noindent\emph{Step 2:} \emph{Merge the branch components.} By merging each of the vertices $v_1\in B$ in all branches of $\mathcal{B}_B(G)=\{\beta_1,\beta_2,\beta_3,\beta_4\}$ and doing the same for the vertex $v_2\in B$, the result is the graph $\mathcal{S}_B(G)$ shown in Figure \ref{Fig:2} (right), which is the specialization of $G$ over the base vertex subset $B$.
\end{example}

Essentially, after selecting a base set $B\subseteq V$ from a graph (network), the specialization model creates a new graph (network) consisting of all the component branches $\mathcal{B}_B(G)$, separated out from one another. Thus a copy of each strongly connected components of $G|_{\bar{B}}$ will appear as a copy in the new graph for each component branch which includes it. In this sense the strongly connected components $C_1,\dots,C_m$ are specialized as the network evolves under this process.

Once a graph has been specialized it can again be specialized by choosing a new base of the specialized graph. In this way a network can be sequentially specialized. As a simple example one can randomly choose a fixed percentage of the graph's vertices at each step (see Figure \ref{Fig:-1}). The result is a graph that has many features consistent with real-world networks. For example, it has a right-skewed degree distribution, is disassociative, has the small world property, is sparse, and its topology is both modular and hierarchical (see Example 3.1 in \cite{BSW18}). As one might expect, the resulting sequence of specializations depends very much on the way in which a base is chosen at each step in the specialization process (cf. Example 3.2 and 3.3 in \cite{BSW18}).

In the following section we consider how the process of specialization effects the spectral properties of a network.

\begin{figure}
  \begin{overpic}[scale=.22]{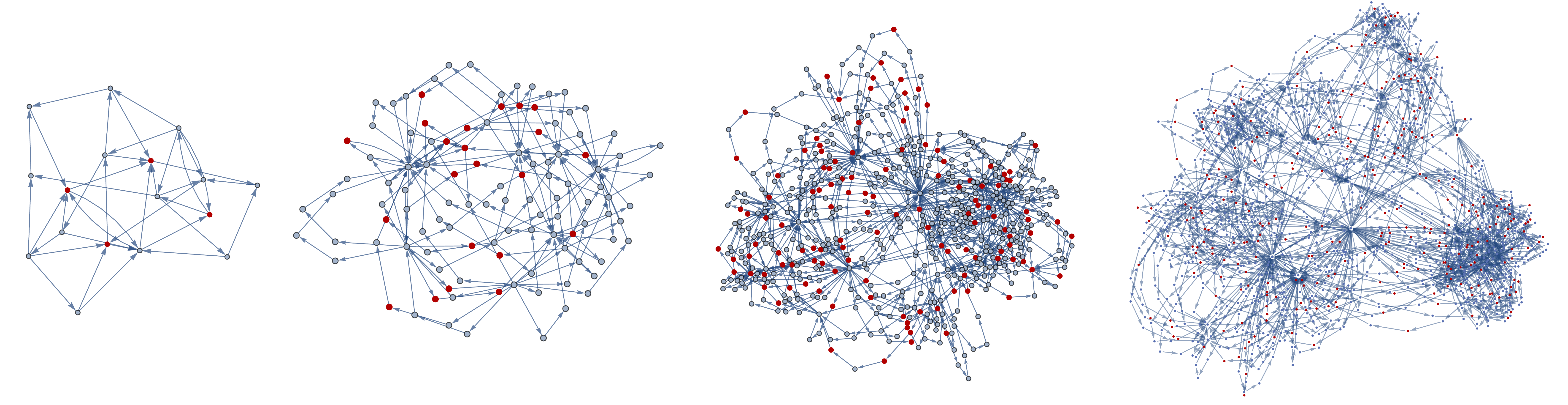}

    \put(8,-1.5){$G_1$}
    \put(30,-1.5){$G_2$}
    \put(55,-1.5){$G_3$}
    \put(85,-1.5){$G_4$}

    \end{overpic}
    \vspace{0.25cm}
  \caption{The unweighted graph $G_1$ is sequentially specialized by randomly choosing eighty percent of its vertices to be its base in each step of the this sequence. The resulting networks $G_1$, $G_2$, $G_3$, and $G_4$ have features that are increasingly similar to real world networks (see Example 3.1 in \cite{BSW18}).}\label{Fig:-1}
\end{figure}
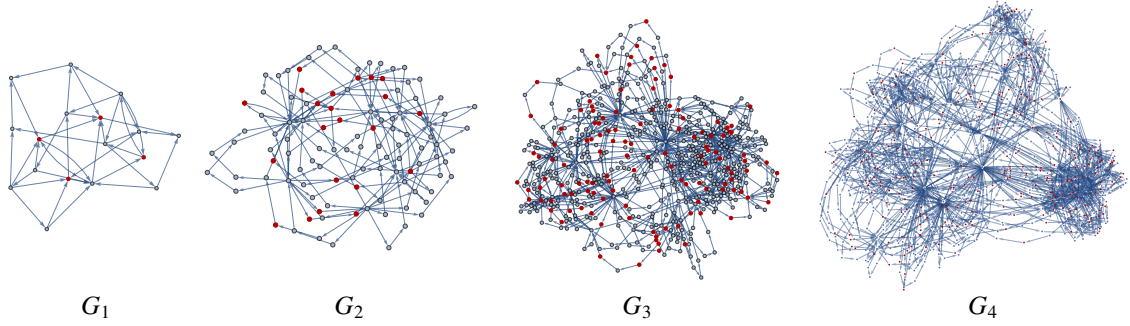

\section{Spectral Properties of Specializations}\label{sec3}
To understand how specializing a network's topology effects the network's dynamics and in turn the network's function, we need some notion that relates both structure and dynamics. One of the most fundamental concepts that relates the two is the notion of a network's spectrum \cite{WM10,BW12,BW13}. Spectral properties are used in a number of network applications including determining network centralities and network communities \cite{Newman10}, which will be important in this section where we describe specializations's effect on a network's eigenvector centrality.

The spectrum of a network can be defined in a number of ways since a number of matrices can be associated with a network. This includes various Laplacian matrices, distance matrices, and adjacency matrices of a graph. The type of matrix we consider here is the weighted adjacency matrix of a graph. The \emph{weighted adjacency matrix} of a graph $G=(V,E,\omega)$ is the matrix $A=\mathcal{A}(G)$ where
\begin{equation}\label{eq:adj}
A_{ij}=
\begin{cases}
\omega(e_{ij}) \ \ \text{if} \ e_{ij}\in E\\
0 \hspace{.8cm} \ \text{otherwise}.
\end{cases}
\end{equation}
If $G$ is unweighted then each entry $A_{ij}=\omega(e_{ij})=1$ if $e_{ij}\in E$ and $A_{ij}=0$ otherwise. The \emph{eigenvalues} of the matrix $A\in\mathbb{R}^{n\times n}$ make up the graph's \emph{spectrum}, which we denote by
\[
\sigma(G)=\{\lambda\in\mathbb{C}:\det(A-\lambda I)=0\},
\]
where we consider $\sigma(G)$ to be a set that includes multiplicities, i.e. a \emph{multiset}. The \emph{spectral radius} of $G$ is the spectral radius of $A=\mathcal{A}(G)$ denoted by
\[
\rho(G)=\max\{|\lambda|:\lambda\in\sigma(A)\}.
\]
In Section \ref{sec4} we will investigate the connection between the spectrum of a graph $G$ and the dynamics of the network associated with it.

Because we are concerned with the spectrum of a graph, which is a set that includes multiplicities, the following is important for our discussion. First, the element $\alpha$ of the set $A$ has \emph{multiplicity} $m$ if there are $m$ elements of $A$ equal to $\alpha$. If $\alpha\in A$ with multiplicity $m$ and $\alpha\in B$ with multiplicity $n$ then\\
\indent (i) the \emph{union} $A\cup B$ is the set in which $\alpha$ has multiplicity $m+n$; and\\
\indent (ii) the \emph{difference} $A-B$ is the set in which $\alpha$ has multiplicity $m-n$ if $m-n>0$ and where $\alpha\notin A-B$ otherwise.

For ease of notation, if $A$ and $B$ are sets that include multiplicity then we let $B^k=\cup_{i=1}^kB$ for $k\geq 1$. That is, the set $B^k$ is $k$ copies of the set $B$ where we let $B^0=\emptyset$. For $k=-1$ we let $A\cup B^{-1}=A-B$. With this notation in place, the spectrum of a graph $G$ and the spectrum of the specialized graph $\mathcal{S}_B(G)$ are related by the following result.

\begin{theorem}\label{thm1} \textbf{(Spectra of Specialized Graphs)}
Let $G=(V,E,\omega)$, $B\subseteq V$, and let $C_1,\dots,C_m$ be the strongly connected components of $G|_{\bar{B}}$. Then
\[
\sigma\big(\mathcal{S}_B(G)\big)=\sigma(G)\cup\sigma(C_1)^{n_1-1}\cup\sigma(C_2)^{n_2-1}\cup\dots\cup \sigma(C_m)^{n_m-1}
\]
where $n_i$ is the number of copies of the component $C_i$ in $\mathcal{S}_B(G)$.
\end{theorem}

If a network has an evolving structure that can be modeled via a graph specialization, or more naturally a sequence of specializations, then Theorem \ref{thm1} allows us to effectively track the changes in the network's spectrum, resulting from the components $C_1,\dots,C_m$. Specifically, the eigenvalues of the resulting specialized graph $G$ are the eigenvalues of the original unspecialized graph together with the eigenvalues of the new copies of the strongly connected components $C_1,\dots,C_m$ appearing in $\mathcal{S}_B(G)$.

In nearly every network $G=(V,E,\omega)$ edges are considered to have \emph{positive edge weights}, i.e. $\omega:E\rightarrow\mathbb{R}^+$ (see [Newman]). In this situation Theorem \ref{thm1} has the following corollary, which is fundamental to the results regarding network dynamics in Section \ref{sec4}.

\begin{corollary} \textbf{(Spectral Radius of Specialized Graphs)}
Suppose $G=(V,E,\omega)$ has positive edge weights. Then for any $B\subseteq V$ the spectral radius $\rho(\mathcal{S}_B(G))=\rho(G)$.
\end{corollary}\label{cor1}

\begin{proof}
Suppose that $G=(V,E,\omega)$ is a graph with positive edge weights and $B\subseteq V$ is a base. Note that as any strongly connected component $C_i$ of $G|_{\bar{B}}$ is a subgraph of $G$ then, as $G$ has positive edge weights, the spectral radius $\rho(C_i)\leq\rho(G)$ (see \cite{HJ90}). It then follows directly from Theorem \ref{thm1} that $\rho(\mathcal{S}_B(G))=\rho(G)$.
\end{proof}

Hence, if a graph has positive weights then its spectral radius is preserved under specialization or any sequence of specializations. The reason is that by construction the edge weights of a graph are preserved when the graph is specialized (see Definition \ref{def:exp}). Hence, if a graph $G$ has positive edge weights then so do any of its specializations or sequential specializations.

We defer the proof of Theorem \ref{thm1} until Section \ref{appendix}. For now, we consider an example of Theorem \ref{thm1}.

\begin{example}\label{ex:3}
The unwieghted graph $G=(V,E)$ in Figure \ref{Fig:2} (left) has eigenvalues
\[\sigma(G)\approx\{2.188,1.167\pm .544i,-1.262\pm0.238i,-1,-1,0\}.\]
Note that for $B=\{v_1,v_2\}$ the graph $G|_{\bar{B}}$ has the strongly connected components $C_1$ and $C_2$ with eigenvalues $\sigma(C_1)=\{2,-1,-1\}$ and $\sigma(C_2)=\{\pm\sqrt{2},0\}$. Since each of the nontrivial branches of $\mathcal{B}_B(G)$ contain one copy of $C_1$ and $C_2$ respectively, Theorem \ref{thm1} implies that the specialized graph $\mathcal{S}_B(G)$ has the twenty-six eigenvalues
\begin{equation}\label{eq:ex2}
\sigma(\mathcal{S}_B(G))=\sigma(G)\cup\sigma(C_1)^3\cup\sigma(C_2)^3.
\end{equation}
Moreover, by inspecting the eigenvalues of $G$, $C_1$, and $C_2$ it follows from Equation \ref{eq:ex2} that $\rho(S_B(G))=\rho(G)\approx2.188$. This follows from Corollary \ref{cor1} as the graph $G$ has unit edge weights, which are positive.
\end{example}

Not only are the eigenvalues and spectral radius of a graph $G$ preserved, in a specific way, as the graph is specialized but so are its eigenvectors. Here an \emph{eigenvector} of a graph $G$ corresponding to the eigenvalue $\lambda\in\sigma(G)$ is a vector $\mathbf{x}$ such that $\mathcal{A}(G)\mathbf{x}=\lambda\mathbf{x}$, in which case $(\lambda, \mathbf{x})$ an \emph{eigenpair} of $G$. That is, the eigenvalues of $G$ are the eigenvalues of its adjacency matrix $A=\mathcal{A}(G)$.

To describe how specialization effects the eigenvectors of a network we require the following definition.

\begin{definition}\label{def:inout}\textbf{(Incoming and Outgoing Component Branches)}
For the graph $G=(V,E,\omega)$ and base $B\subseteq V$ let $\beta=\{v_i,e_{0},C_1,e_{1},C_2,\dots,C_m,e_{m},v_{j}\}\in\mathcal{B}_B(G)$. We call the ordered set
\[
In(\beta,C_k)=\{v_i,e_{0},C_1,e_{1},C_2,\dots,C_k\}\subset \beta
\]
the \emph{incoming branch} of $\beta$ up to $C_k$. Similarly, we call the ordered set
\[
Out(\beta,C_k)=\{C_k,e_{k},C_{k+1},\dots,C_m,e_{m},v_{j}\}\subset \beta
\]
the \emph{outgoing branch} in $\beta$ from $C_k$.
\end{definition}

If $Z$ is a strongly connected component of $G|_{\bar{B}}$ then $\ell\geq 0$ copies of it will appear in the graph $\mathcal{S}_B(G)$, which we denote by $\mathcal{C}(Z)=\{Z_1,Z_2,\dots,Z_{\ell}\}$. Here, each $Z_i$ is associated with the component $Z$ in a single branch $\beta_i\in\mathcal{B}_B(G)$. We say $Z_i,Z_j\in\mathcal{C}(Z)$ have the same incoming branch if $In(\beta_i,Z)=In(\beta_j,Z)$ and the same outgoing branch if $Out(\beta_i,Z)=Out(\beta_j,Z)$.

\begin{definition}\textbf{(Eigenvector Transfer Matrix)}
For $\beta=\{v_i,e_{0},C_1,e_{1},C_2,\dots,C_m,e_{m},v_{j}\}\in\mathcal{B}_B(G)$ let the adjacency matrix of the incoming branch $In(\beta,C_k)$ together with $G|_B$ be the matrix
\begin{equation*}\label{ex:system}
A=\begin{bmatrix}
\underline{B}&&&&\\
Y_0&\underline{C}_1&&&\\
&Y_1&\underline{C}_2&&\\
&&\ddots&\ddots&\\
&&&Y_{k-1}&\underline{C}_k
\end{bmatrix}
\end{equation*}
where $\underline{B}=\mathcal{A}(G|_{B})$ and $\underline{C}_i=\mathcal{A}(C_i)$. We call the matrix
\[
T(\beta,C_k,\lambda)=(\lambda I-\underline{C}_k)^{-1}Y_{k-1}(\lambda I-\underline{C}_{k-1})^{-1}Y_{k-2}\cdots(\lambda I-\underline{C}_1)^{-1}Y_0
\]
the \emph{eigenvector transfer matrix} of $In(\beta,C_k)$, where $\lambda$ is a spectral parameter.
\end{definition}

If $\mathbf{x}$ is an eigenvector of $G$ and $S\subseteq V$ is any subset of its vertex set, we let $\mathbf{x}_S$ denote the vector $\mathbf{x}$ restricted to the entries indexed by $S$. Similarly, by slight abuse of this notation, if $Z$ is a subgraph of $G$ with vertex set $S$ we let $\mathbf{x}_Z=\mathbf{x}_S$. With this in place we state the following theorem describing how the eigenvectors of a graph are effected by specialization.

\begin{theorem}\label{prop:0}\textbf{(Eigenvectors of Specialized Graphs)}
Let $G=(V,E,\omega)$ be a graph, $B\subseteq V$ a base, and let $Z$ be a strongly connected component of $\beta\in\mathcal{B}_B(G)$. If $(\lambda,\mathbf{u})$ is an eigenpair of $G$ with $\lambda\notin\sigma(G|_{\bar{B}})$ then there is an eigenpair $(\lambda,\mathbf{v})$ of $\mathcal{S}_B(G)$ such that the following hold:\\
(i) $\mathbf{u}_B=\mathbf{v}_B$.\\
(ii) For all $Z_i\in\mathcal{C}(Z)$ the eigenvector restriction
\[
\mathbf{v}_{Z_i}=T(\beta,Z,\lambda)\mathbf{v}_B.
\]
Hence, if $Z_i,Z_j\in\mathcal{C}(Z)$ have the same incoming branch then $\mathbf{v}_{Z_i}=\mathbf{v}_{Z_j}$.\\
(iii) For $Z_i\in\mathcal{C}(Z)$ let $\cup_{k=1}^\ell\{Z_k\}$ be the copies of $Z$ that have the same outgoing branch as $Z_i$. Then
\[
\mathbf{u}_Z=\sum_{k=1}^{\ell}\mathbf{v}_{Z_k}=\sum_{k=1}^{\ell}T(\beta,Z,\lambda)\mathbf{v}_B.
\]
\end{theorem}

Part (i) of Theorem \ref{prop:0} states that the graphs $G$ and $\mathcal{S}_B(G)$ have the same eigenvectors if we restrict our attention to those entries that correspond to the base vertices $B$ and to those eigenvectors with eigenvalues in $\sigma(G)-\sigma(G|_{\bar{B}})\subset\sigma(\mathcal{S}_B(G))$. Part (ii) states that the eigenvectors associated with the component $Z_i$ can be found by applying the transfer matrix $T(\beta,Z,\lambda)$ to the  same eigenvector associated with $G|_B$. Hence, if two specializations $Z_i$ and $Z_j$ of the same component $Z$ have the same incoming branch, their associated eigenvectors are identical. Part (iii) states that if we sum the eigenvectors associated of each copy $\cup_{k=1}^\ell\{Z_k\}$ of $Z$ with the same outgoing branch in $\mathcal{S}_B(G)$ then we recover the eigenvector associated with the original component $Z$ in $G$.

\begin{example}
Consider the unweighted graph $G=(V,E)$ in Figure \ref{Fig:3} and its specialization over the base $B=\{v_1,v_2\}$, shown in red, as in Example \ref{ex:2}. Here, for simplicity, we let the components $C_1=Y$ and $C_2=Z$, in which case $\mathcal{C}(Y)=\{Y_1,\dots,Y_4\}$ and $\mathcal{C}(Z)=\{Z_1,\dots,Z_4\}$ are the copies of $Y$ and $Z$, respectively, in $\mathcal{S}_B(G)$. Note that
\[
\mathcal{C}(Y)=\{\{Y_1,Y_2\},\{Y_3,Y_4\}\} \ \text{ and } \ \mathcal{C}(Z)=\{\{Z_1,Z_2\},\{Z_3,Z_4\}\}
\]
are the partition of $\mathcal{C}(Y)$ and $\mathcal{C}(Z)$ into components with the same incoming branches and
\[
\mathcal{C}(Y)=\{\{Y_1,Y_3\},\{Y_2,Y_4\}\} \ \text{ and } \ \mathcal{C}(Z)=\{\{Z_1,Z_3\},\{Z_2,Z_4\}\}
\]
are the partition of these sets into components with the same outgoing branches.

In Figure \ref{Fig:3}, numbers next to vertices in both $G$ and $\mathcal{S}_B(G)$ are the associated entries of the graph's eigenvectors $\mathbf{u}$ and $\mathbf{v}$ corresponding to the spectral radius $\rho$ of both graphs, respectively. Since $\rho=\rho(G)\approx 2.188$ is an eigenvalue of $G$ that is not an eigenvalue of $G|_{\bar{B}}$, part (i) of Theorem \ref{prop:0} implies that $\mathbf{u}_B=[1,1]^T=\mathbf{u}_B$ as can be seen in the figure. Also, as $Z_1$ and $Z_2$ have the same incoming branch in the specialized graph, part (ii) of Theorem \ref{prop:0} implies that $\mathbf{v}_{Z_1}=[1.3,.6,.6]^T=\mathbf{v}_{Z_2}$. The same can be quickly verified for the other components with the same incoming branches. Last, since $Z_1$ and $Z_3$ are the copies of $Z$ with the same outgoing branches then part (iii) of Theorem \ref{prop:0} implies
\[
\mathbf{u}_Z=[2.6,1.2,1.2]^T=[1.3,.6,.6]^T+[1.3,.6,.6]^T=\mathbf{v}_{Z_1}+\mathbf{v}_{Z_3}
\]
and the same can be verified for the other components with identical outgoing branches. Additionally, the eigenvector transfer matrix can be used to determine any of the component's associated eigenvectors in the specialized graph using only this matrix and the vector $\mathbf{v}_B=[1,1]^T$ so long as the eigenvalues associated with the eigenvector is not an eigenvalue of $G|_{\bar{B}}$.
\end{example}

\begin{figure}
  \begin{overpic}[scale=.28]{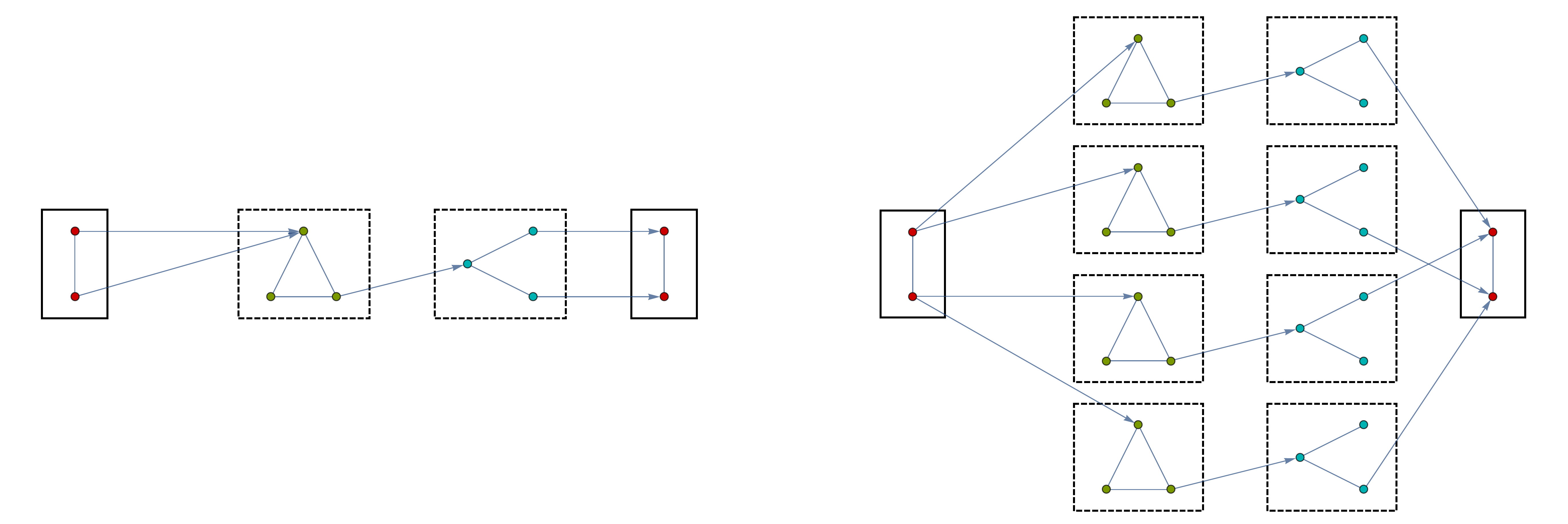}
    \put(15,3){Unspecialized Graph $G$}
    \put(67,-5){Specialized Graph $\mathcal{S}_B(G)$}

    \put(4,8){$B$}
    \put(4.25,20.5){\tiny{\textbf{1}}}
    \put(4.25,11.5){\tiny{\textbf{1}}}

    \put(15.5,8){$C_1=Y$}
    \put(18.25,20.5){\tiny{\textbf{3.94}}}
    \put(16,11.5){\tiny{\textbf{3.31}}}
    \put(20,11.5){\tiny{\textbf{3.31}}}

    \put(28.5,8){$C_2=Z$}
    \put(33,20.5){\tiny{\textbf{1.2}}}
    \put(29,14.5){\tiny{\textbf{2.6}}}
    \put(33,11.5){\tiny{\textbf{1.2}}}

    \put(41.5,8){$B$}
    \put(41.75,20.5){\tiny{\textbf{1}}}
    \put(41.75,11.5){\tiny{\textbf{1}}}

    \put(57,8){$B$}
    \put(57.5,20.5){\tiny{\textbf{1}}}
    \put(57.5,11.5){\tiny{\textbf{1}}}

    \put(73.5,31){\tiny{\textbf{1.97}}}
    \put(69,27.5){\tiny{\textbf{1.65}}}
    \put(73.5,27.5){\tiny{\textbf{1.65}}}

    \put(73.5,22.75){\tiny{\textbf{1.97}}}
    \put(69,19.25){\tiny{\textbf{1.65}}}
    \put(73.5,19.25){\tiny{\textbf{1.65}}}

    \put(73.5,14.5){\tiny{\textbf{1.97}}}
    \put(69,11){\tiny{\textbf{1.65}}}
    \put(73.5,11){\tiny{\textbf{1.65}}}

    \put(73.5,6.25){\tiny{\textbf{1.97}}}
    \put(69,2.75){\tiny{\textbf{1.65}}}
    \put(73.5,2.75){\tiny{\textbf{1.65}}}

    \put(86,29.5){\tiny{\textbf{.6}}}
    \put(81.5,29.5){\tiny{\textbf{1.3}}}
    \put(86,27.5){\tiny{\textbf{.6}}}

    \put(86,21.25){\tiny{\textbf{.6}}}
    \put(81.5,21.25){\tiny{\textbf{1.3}}}
    \put(86,19.25){\tiny{\textbf{.6}}}

    \put(86,13){\tiny{\textbf{.6}}}
    \put(81.5,13){\tiny{\textbf{1.3}}}
    \put(86,11){\tiny{\textbf{.6}}}

    \put(86,4.75){\tiny{\textbf{.6}}}
    \put(81.5,4.75){\tiny{\textbf{1.3}}}
    \put(86,2.75){\tiny{\textbf{.6}}}

    \put(94,8){$B$}
    \put(94.5,20.5){\tiny{\textbf{1}}}
    \put(94.5,11.5){\tiny{\textbf{1}}}

    \put(65.5,28.5){$Y_1$}
    \put(65.5,20){$Y_2$}
    \put(65.5,12){$Y_3$}
    \put(65.5,3.5){$Y_4$}

    \put(90,28.5){$Z_1$}
    \put(90,20){$Z_2$}
    \put(90,12){$Z_3$}
    \put(90,3.5){$Z_4$}


    \end{overpic}
    \vspace{0.5cm}
\caption{The graph $G$ and its specialization $\mathcal{S}_B(G)$, first considered in Example \ref{ex:2} are shown where we let $C_1=Y$ and $C_2=Z$. Here $\mathcal{C}(Y)=\{Y_1,\dots,Y_4\}$ and $\mathcal{C}(Z)=\{Z_1,\dots,Z_4\}$ are the copies of $Y$ and $Z$ in $\mathcal{S}_B(G)$, respectively. Numbers next to vertices are the vertices' eigenvector centrality, equivalently the corresponding entry in the vertices' eigenvector centrality.}\label{Fig:3}
\end{figure}

One immediate consequence of Theorem \ref{prop:0} is that we can describe to a large extent what happens to the eigenvector centrality of a network as the network becomes increasingly specialized. Eigenvector centrality is one of the standard network measures used for determining the importance of a vector in a network (see \cite{Newman10} for more details). For an unweighted graph $G=(V,E)$ that is strongly connected the graphs eigenvector $\mathbf{p}$ corresponding to its spectral radius gives the relative ranking $p_i$ to each vertex $v_i\in V$. This value $p_i$ is referred to as the \emph{eigenvector centrality} of the vertex $v_i$. Here we refer to the vector $\mathbf{p}$ as an \emph{eigencentrality vector} of the graph $G$.

\begin{corollary}\label{cor:2}\textbf{(Eigenvector Centrality of Specialized Networks)}
Let $G=(V,E)$ be an unweighted strongly connected graph with base $B\subseteq V$ and $Z$ a strongly connected component of $G|_{\bar{B}}$ . If $\mathbf{p}$ is the eigencentrality vector of $G$, there exists an eigencentrality vector $\mathbf{q}$ of $\mathcal{S}_B(G)$ such that\\
(i) $\mathbf{p}_B=\mathbf{q}_B$;\\
(ii) if $Z_i,Z_j\in\mathcal{C}(Z)$ have the same incoming branch then $\mathbf{q}_{Z_i}=\mathbf{q}_{Z_j}$; and\\
(iii) if $\cup_{k=1}^\ell\{Z_k\}$ are the copies of $Z$ that have the same outgoing branch as some $Z_i\in\mathcal{C}(Z)$ then $\mathbf{p}_Z=\sum_{k=1}^{\ell}\mathbf{q}_{Z_k}$.
\end{corollary}

\begin{proof}
By convention, the unweighted graph $G=(V, E)$ has positive edge weights as each edge is assigned unit weight. Since $G$ is strongly connected, the Perron-Frobenius theorem implies that the spectral radius $\rho=\rho(G)$ is a simple eigenvalue of $G$ with corresponding eigencentrality vector $\mathbf{p}$. Since $G|_{\bar{B}}$  is a subgraph of the strongly connected graph $G$ then the spectral radius $\rho(G|_{\bar{B}})<\rho(G)$ (see \cite{HJ90}). Hence, $\rho\notin\sigma(G|_{\bar{B}})$. Part (i) of Theorem \ref{prop:0} then implies that there is an eigenpair $(\rho, \mathbf{q})$ of $\mathcal{S}_B(G)$ such that $\mathbf{p}_B = \mathbf{q}_B$.

Note that $S_B(G)$ is both unweighted and strongly connected as $G$ is unweighted and strongly connected. By Corollary \ref{cor1} both $G$ and $\mathcal{S}_B(G)$ have the same spectral radius $\rho$. Therefore, $\mathbf{q}$ must be an eigencentrality
vector of $\mathcal{S}_B(G)$, completing the proof of part (i).

Parts (ii) and (iii) then follow immediately from parts (ii) and (iii) of Theorem \ref{prop:0}.
\end{proof}

When a network is specialized the eigenvector centrality of the unspecialized vertices, i.e. the vertices in the base, stay the same. In contrast, this process does effect the centrality of those vertices that are specialized (cf. Figure \ref{Fig:3}). What we can say for these vertices is that if two vertices are copies of the same vertex that belong to components with the same incoming branches then these vertices have the same eigenvector centrality. Also, if we take all the specializations, i.e. copies, of a vertex that belong to components with the same outgoing branch then summing these vertices' centralities results in the centrality of the original vertex. That is, the centrality of a vertex is distributed among this set of its copies (cf. Figure \ref{Fig:3}).

As with the proof of Theorem \ref{thm1}, the proof of Theorem \ref{prop:0} is given in Section \ref{appendix}. These proofs rely on a number of results from the theory of isospectral graph reductions found in \cite{BW12,BWBook}, which can also be found in Section \ref{appendix}.

\section{Network Growth, Dynamics, and Function}\label{sec4}
Sections \ref{sec2} and \ref{sec3} of this paper are primarily concerned with the \emph{dynamics of} a network, i.e. the temporal evolution of the network's structure of interactions, and the spectral consequences of this evolution. This evolution also effects the network's ability to perform its intended function. This function is not only dependent on the network's topology but also on the type of dynamics that emerges from the interactions between the network elements, i.e. the dynamics \emph{on} the network. For instance, power is transferred efficiently in power grids when the grid is synchronized.

In this section our goal is to understand how the dynamics \emph{of} a network can impact the dynamics \emph{on} the network. Specifically, we study under what condition(s) a network can maintain its dynamics as the network's structure evolves via specialization. This ability to maintain functionality even as the network grows is observed in many systems. A prime example is the physiological network of organs within the body, all of which develop over time but maintain specific functions \cite{Plamen2015}.

The dynamics \emph{on} a network with a fixed structure can be formalized as follows.

\begin{definition}\label{def:dn}{\textbf{\emph{(Dynamical Network)}}}
Let $(X_i,d_i)$ be a complete metric space where $X_i\subseteq \mathbb{R}$ and let $(X,d_{max})$ be the complete metric space formed by giving the product space $X=\bigoplus_{i=1}^n X_i$ the metric
\[
d_{max}(\mathbf{x},\mathbf{y}) = \max_i d_i(x_i,y_i) \quad \text{ where } \quad \mathbf{x},\mathbf{y} \in X \quad \text{ and } \quad x_i,y_i \in X_i.
\]
Let $F:X \to X$  be a $C^1(X)$ map with $i^{th}$ component function $F_i:X\to X_i$ given by
\[
F_i = \bigoplus_{j\in I_i}X_j\rightarrow X_i \quad \text{where} \quad I_i\subseteq N=\{1,2,\dots,n\}.
\]
The discrete-time dynamical system $(F,X)$ generated by iterating the function $F$ on $X$ is called a \emph{dynamical network}. If an initial condition $\mathbf{x}^0\in X$ is given, we define the $k^{th}$ \emph{iterate} of $\mathbf{x}^0$ as $\mathbf{x}^k=F^k(\mathbf{x}^0)$, with orbit $\{F^k(\mathbf{x}^0)\}_{k=0}^\infty=\{\mathbf{x}^0,\mathbf{x}^1,\mathbf{x}^2,\hdots\}$ in which $\mathbf{x}^k$ is the state of the network at time $k \ge 0$.
\end{definition}

The component function $F_i$ describes the dynamics of the $i^{th}$ network element that emerges from its interactions with a subset of the other network elements where there is a directed interaction between the $i^{th}$ and $j^{th}$ elements if $j\in I_i$. For the initial condition $\mathbf{x}^0\in X$ the state of the $i^{th}$ element at time $k\ge 0$ is $x^k_i=(F^k(\mathbf{x}^0))_i\in X_i$ where $X_i$ is the state space of the $i^{th}$ element. The state space $X=\bigoplus_{i=1}^n X_i$ is the collective state space of all network elements.

For simplicity in our discussion we assume that the map $F:X\rightarrow X$ is continuously differentiable and that each $X_i$ is some closed interval of real numbers (see Definition \ref{def:dn}). In fact, all of our results in this section hold in the more general setting where $F:X\rightarrow X$ is Lipschitz continuous and $X$ is any complete metric space (see \cite{BW12,BWBook}).

The specific type of dynamics we consider here is global stability, which is observed in a number of systems including neural networks \cite{Cao2003,Cheng2006,SChena2009,MCohen1983,LTao2011}, in epidemic models \cite{Wang2008}, and is also important in the study of congestion in computer networks \cite{Alpcan2005}. In a globally stable network, which we will simply refer to as \emph{stable}, the state of the network tends towards an equilibrium irrespective of its present state. Formally, network stability is defined as follows.

\begin{definition}\textbf{(Network Stability)}
The dynamical network $(F,X)$ is \emph{globally stable} if there is an $\mathbf{x}^*\in X$ such that for any $\mathbf{x}^0\in X$
\[
\lim_{k\rightarrow\infty}d_{max}(F^k(\mathbf{x}^0),\mathbf{x}^*)=0.
\]
\end{definition}

A globally attracting equilibrium $\mathbf{x}^*$ in a network is presumably a state at or near which a network can efficiently carry out its function. Whether or not this equilibrium remains stable over time depends on a number of factors including external influences such as changes in the network's environment, etc. However, not only can outside influences destabilize a network but potentially the network's own growth can have this effect. As mentioned in the introduction, an important example of this type of destabilization is cancer, which is the abnormal growth of cells that can lead to significant problems in biological networks.

Here we consider how growth via specialization can effect the stability of a network. Specifically, we consider the specialization of a class of dynamical networks $(F,X)$ with components of the form
\begin{equation}\label{eq:netclass}
F_i(\mathbf{x})=\sum_{j=1}^n A_{ij}f_{ij}(x_j), \quad \text{for} \quad i\in N=\{1,2,\dots,n\}
\end{equation}
where the matrix $A\in\{0,1\}^{n\times n}$ is a matrix of zeros and ones and each $f_{ij}:X_j\rightarrow\mathbb{R}$ are $C^1(X_j)$ functions with bounded derivatives for all $i,j\in N$. We refer to the graph $G$ with adjacency matrix $A=\mathcal{A}(G)$ in Equation \eqref{eq:netclass} as the \emph{graph of interactions} of $(F,X)$.

It is worth noting that we could absorb the matrix $A$ into the functions $f_{ij}$. However, we use this matrix as a means of specializing the dynamical network $(F,X)$ in a way analogous to the method of specialization described in Section \ref{sec2} for graphs. This is possible as there is a one-to-one relation between a graph $G=(V,E,\omega)$ and its weighted adjacency matrix $\mathcal{A}(G)\in\mathbb{R}^{n\times n}$. Therefore, we can use the notion of a graph specialization to define a matrix specialization.

\begin{definition}\label{def:matspec}\textbf{(Matrix Specialization)}
Let $A\in\mathbb{R}^{n\times n}$ and $B\subseteq N=\{1,2,\dots,n\}$ be a base. Then the \emph{specialization} of $A$ over $B$ is the matrix
\[
\underline{A}=\mathcal{S}_B(A)=\mathcal{A}(\mathcal{S}_B(G))\in\mathbb{R}^{m\times m}
\]
where $A=\mathcal{A}(G)$. Additionally, suppose $G=(V,E,\omega)$ and $\mathcal{S}_B(G)=(\mathcal{V},\mathcal{E},\mu)$. For $M=\{1,2,\dots,m\}$ let $\tau:M\rightarrow N$ where $\tau(i)=j$ if $\nu_i\in\mathcal{V}$ is a copy of $v_j\in V$. We refer to the function $\tau$ as the \emph{origination function} of this specialization.
\end{definition}

Note that we are slightly abusing notation in Definition \ref{def:matspec} by letting the base $B$ be both a subset of $N=\{1,2,\dots,n\}$ and a subset of $V=\{v_1,v_2,\dots,v_n\}$. The idea is that $B\subseteq N$ is a set of indices over which the matrix $A$ is specialized, which in turn is the set that indexes the vertices $B\subseteq V$ over which $G=(V,E,\omega)$ is specialized. Roughly speaking, to specialize the matrix $A\in\mathbb{R}^{n\times n}$ we specialize the associated graph $G$ with adjacency matrix $A$. The adjacency matrix of the resulting specialized graph is the specialization of the matrix $A$. This allows us to specialize dynamical networks as follows.

\begin{definition}\label{def:specdyn}\textbf{(Specializations of Dynamical Networks)}
Suppose $(F,X)$ is a dynamical network given by Equation \eqref{eq:netclass}. If $B\subseteq \{1,2,\dots,n\}$ then the \emph{specialization} of $(F,X)$ over the base $B$ is the dynamical network $(G,Y)$ with components
\[
G_i(\mathbf{y})=\sum_{j=1}^m \underline{A}_{ij}f_{\tau(i)\tau(j)}(y_j), \quad \text{for} \quad i\in M=\{1,2,\dots,m\}
\]
where $\underline{A}=\mathcal{S}_B(A)$, $Y=\bigoplus_{j=1}^m Y_j$ with $Y_j=X_{\tau(j)}$, $y_j=x_{\tau(j)}$, and $\tau:M\rightarrow N$ is the origination function of the specialization.
\end{definition}

To give an example of a well-studied class of dynamical networks that can be specialized according to Definition \ref{def:specdyn} consider the class of dynamical networks known as discrete-time recurrent neural networks. The stability of such systems has been the focus of a large number of studies \cite{Cheng2006,SChena2009,MCohen1983,LTao2011}, especially time-delayed versions of these systems \cite{LWSL07}, both of which have the form given in Equation \eqref{eq:netclass}.

\begin{definition}\label{def:DRNN}\textbf{(Discrete-Time Recurrent Neural Network)}
A \emph{discrete-time recurrent neural network} $(R,X)$ is a dynamical network of the form
\begin{equation}\label{eq:DRNN}
R_i(\mathbf{x})=a_i x_i+\sum_{j=1,j\neq i}^n W_{ij}g_j(x_j)+c_i, \ \ i\in N=\{1,2,\dots,n\}.
\end{equation}
where the component $R_i$ describes the dynamics of the $i^{th}$ neuron in which each $|a_i|<1$ are the \emph{feedback coefficients}, the matrix $W\in\mathbb{R}^{n\times n}$ with $W_{ii}=0$ is the \emph{connection weight matrix}, and the constants $c_i$ are the \emph{exogenous inputs} to the network.
\end{definition}

In the general theory of discrete-time recurrent neural networks (DRNN) the functions $g_j:\mathbb{R}\rightarrow\mathbb{R}$ are typically assumed to be differentiable, monotonically increasing, and bounded. Here, for the sake of illustration, we make the additional assumption that each $g_j$ has a bounded derivative.


In the following example we consider the specialization of a DRNN and the dynamic consequences of this specialization.

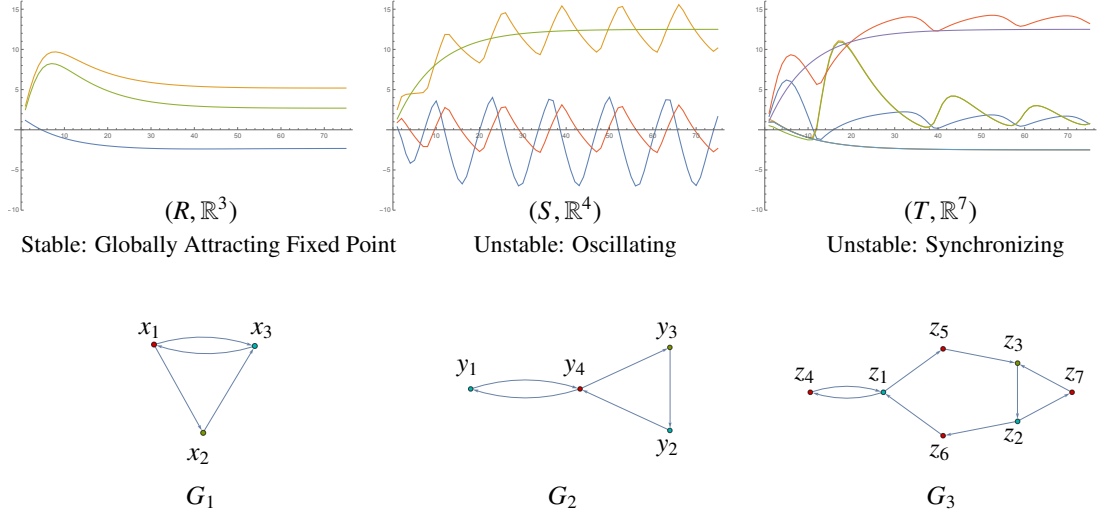
\begin{figure}
\begin{center}
\begin{tabular}{c}
    \begin{overpic}[scale=.25]{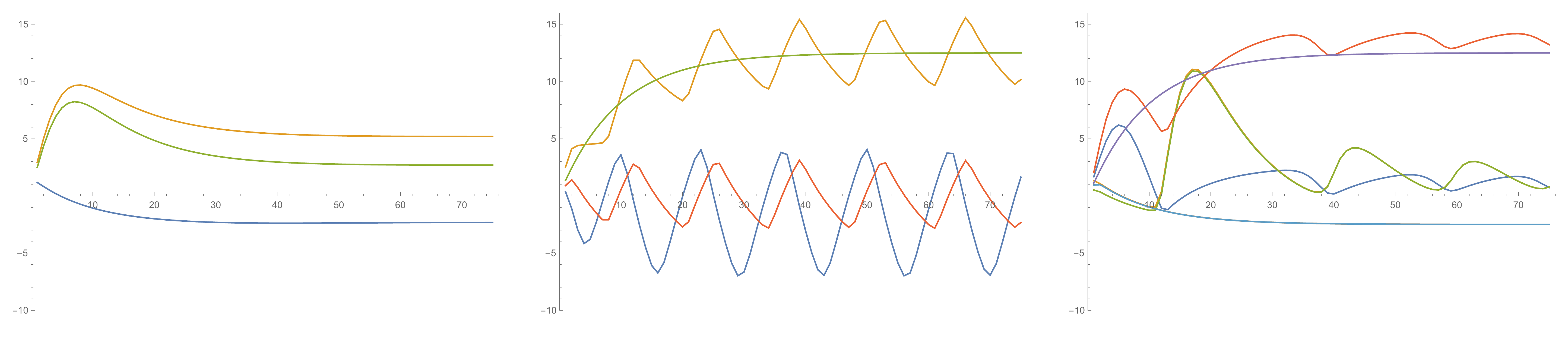}
    \put(2,-2){\small{Stable: Globally Attracting Fixed Point}}
    \put(15,1){$(R,\mathbb{R}^3)$}
    \put(43,-2){\small{Unstable: Oscillating}}
    \put(48,1){$(S,\mathbb{R}^4)$}
    \put(75,-2){\small{Unstable: Synchronizing}}
    \put(82,1){$(T,\mathbb{R}^7)$}
    \end{overpic}\\\\\\
    \begin{overpic}[scale=.18]{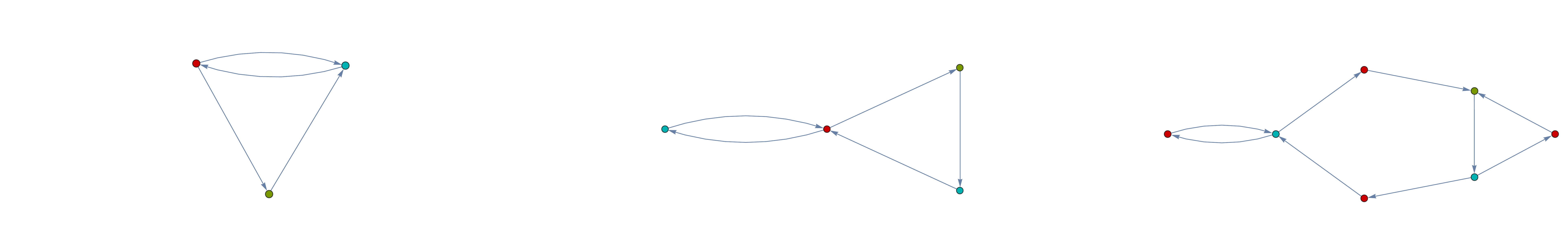}
    \put(15.6,-3){$G_1$}
    \put(21.6,13){$x_3$}
    \put(15.6,1.1){$x_2$}
    \put(11,13){$x_1$}

    \put(49.75,-3){$G_2$}
    \put(41.1,9.5){$y_1$}
    \put(60,2){$y_2$}
    \put(60,13){$y_3$}
    \put(51.25,9.5){$y_4$}

    \put(85.5,-3){$G_3$}
    \put(80,8.75){$z_1$}
    \put(92.75,3){$z_2$}
    \put(92.75,11.75){$z_3$}
    \put(73,8.75){$z_4$}
    \put(85.75,13){$z_5$}
    \put(85.75,1.5){$z_6$}
    \put(98.5,8.75){$z_7$}
    \end{overpic}\\\\
\end{tabular}
\caption{Left: The stable dynamics of the network $(R,\mathbb{R}^3)$ considered in Example \ref{ex:DRNN} is shown (top) together with its graph of interactions $G_1$ (bottom). Center: The unstable periodic dynamics of the specialization $(S,\mathbb{R}^4)$ of $(R,\mathbb{R}^3)$ over the set $B=\{1,2\}$ is shown (top) with its graph of interactions $G_2$ (bottom). Right: The unstable but partially synchronizing dynamics of the specialization $(T,\mathbb{R}^7)$ of $(S,\mathbb{R}^4)$ over the set $C=\{1,2,3\}$ is shown (top) together with its graph of interactions $G_3$. In each of $G_1$, $G_2$, and $G_3$ loops are omitted and vertex colors indicate from which element of the previous network an element was specialized.}\label{Fig:Spec1}
\end{center}
\end{figure}

\begin{example}\label{ex:DRNN}\textbf{(Destabilization via Specialization)}
Consider the dynamical network $(R,\mathbb{R}^3)$ given by
\begin{equation}\label{ex:system}
R\left(\begin{bmatrix}
x_1 \\
x_2 \\
x_3
\end{bmatrix}\right)
=
\left[
\begin{array}{c}
\frac{9}{10}x_1-2\tanh(x_3)+\frac{7}{4}\\ [1pt]
\frac{9}{10}x_2+\tanh(x_1)+\frac{5}{4}\\ [1pt]
\frac{9}{10}x_3+\tanh(x_1)+\tanh(x_2)+\frac{2}{4}
\end{array}\right],
\end{equation}
which is the DRNN in which
\[
f_{ij}(x)=
\begin{cases}
g_i(x)=\tanh(x)  &\text{if} \ i \neq j\\
a_i x +c_i &\text{otherwise}
\end{cases}
\]
where $a_i=\frac{9}{10}$ for $i=1,2,3$; $c_1=\frac{7}{4}$, $c_2=\frac{5}{4}$, $c_3=\frac{2}{4}$; with the connection weight matrix
\[
W=
\left[\begin{array}{ccc}
0&0&-2\\
1&0&0\\
1&1&0
\end{array}\right]
\]
and $A=W+\frac{9}{10}I$. We choose the function $g_i(x)=\tanh(x)$ as this is a standard activation function used to model neural interactions in network science. The network's graph of interactions $G_1$ is shown in Figure \ref{Fig:Spec1} (bottom left). As is shown in Figure \ref{Fig:Spec1} (top left) the dynamical network $(R,\mathbb{R}^3)$ is stable with globally attracting fixed point $\mathbf{x}^*=(-2.49,2.65,5.05)$.

If the network $(R,\mathbb{R}^3)$ is specialized over the base $B=\{1,2\}$ the result is the network $(S,\mathbb{R}^4)$ given by
\begin{equation*}
S\left(\begin{bmatrix}
y_1 \\
y_2 \\
y_3 \\
y_4
\end{bmatrix}\right)
=
\left[
\begin{array}{c}
\frac{9}{10}y_1+\tanh(y_4)+\frac{2}{4}\\ [1pt]
\frac{9}{10}y_2+\tanh(y_3)+\frac{2}{4}\\ [1pt]
\frac{9}{10}y_3+\tanh(y_4)+\frac{5}{4}\\ [1pt]
\frac{9}{10}y_4-2\tanh(y_1)-2\tanh(y_2)+\frac{7}{4}\\
\end{array}\right].
\end{equation*}
The network's graph of interaction $G_2=\mathcal{S}_B(G_1)$ is shown in Figure \ref{Fig:Spec1} (top center) where the vertex colors indicate from which network element the elements of $G_2$ were specialized. Although the dynamics of the original network $(R,\mathbb{R}^3)$ is stable, the dynamics of its specialization $(S,\mathbb{R}^4)$ oscillates periodically as can be seen in Figure \ref{Fig:Spec1} (bottom center). That is, specialization does not, at least in general, preserve stability.

The network can again be specialized over the base $C=\{1,2,3\}$ which results in the network of seven elements $(T,\mathbb{R}^7)$ given by
\begin{equation*}
T\left(\begin{bmatrix}
z_1 \\
z_2 \\
z_3 \\
z_4 \\
z_5 \\
z_6 \\
z_7
\end{bmatrix}\right)
=
\left[
\begin{array}{c}
\frac{9}{10}z_1+\tanh(z_4)+\tanh(z_6)+\frac{2}{4}\\ [1pt]
\frac{9}{10}z_2+\tanh(z_3)+\frac{2}{4}\\ [1pt]
\frac{9}{10}z_3+\tanh(z_5)+\tanh(z_7)+\frac{5}{4}\\ [1pt]
\frac{9}{10}z_4-2\tanh(z_1)+\frac{7}{4}\\ [1pt]
\frac{9}{10}z_5-2\tanh(z_1)+\frac{7}{4}\\ [1pt]
\frac{9}{10}z_6-2\tanh(z_2)+\frac{7}{4}\\ [1pt]
\frac{9}{10}z_7-2\tanh(z_2)+\frac{7}{4}
\end{array}\right].
\end{equation*}
The network's graph of interaction $G_3=\mathcal{S}_C(G_2)$ is shown in Figure \ref{Fig:Spec1} (bottom right) where, as before, the vertex colors indicate from which network element the elements of $G_3$ were specialized. Here specialization has again altered the dynamics of the network as $(T,\mathbb{R}^{7})$ has unstable but now synchronizing dynamics. Specifically, elements $z_4$ and $z_6$ and also elements $z_5$ and $z_7$ synchronize irrespective of the network's initial condition (see Figure \ref{Fig:Spec1} top right).
\end{example}

In previous studies of dynamical networks the goal has been to determine under what condition(s) a given network has stable dynamics (see for instance the references in \cite{LWSL07}). Here we consider a different but related question which is, under what condition(s) does a dynamical network with an evolving structure of interactions maintain its stability as it evolves.

As a partial answer to this general question, we use the following notion of a stability matrix, which allows us to study the change or lack of change in a network's stability after it has been specialized.

\begin{definition} \textbf{(Stability Matrix)}
For the dynamical network $(F,X)$ suppose there exist finite constants
\begin{equation}\label{eq:stability}
\Lambda_{ij}=\sup_{\mathbf{x}\in X}\left|\frac{\partial F_i}{\partial x_j}(\mathbf{x})\right|<\infty \ \text{for all} \ i,j\in N=\{1,2,\dots,n\}.
\end{equation}
Then we call the matrix $\Lambda\in\mathbb{R}^{n\times n}$ the \emph{stability matrix} of $(F,X)$.
\end{definition}

The stability matrix $\Lambda$ can be thought of as a global linearization of the typically nonlinear dynamical network $(F,X)$. The following result states that if the eigenvalues of the matrix $\Lambda$ lie within the unit circle then the dynamical network $(F,X)$ is stable, the proof of which can be found in \cite{BW12}.

\begin{theorem}\label{stability} \textbf{(Network Stability)}
Suppose $\Lambda$ is the stability matrix of the dynamical network $(F,X)$. If $\rho(\Lambda)<1$ then the dynamical network $(F,X)$ is stable.
\end{theorem}

As it will be helpful in what follows, if $\Lambda$ is the stability matrix of the dynamical network $(F,X)$ we let $\rho(F)=\rho(\Lambda)$ be the \emph{spectral radius} of the network.

A key feature of the stability described in Theorem \ref{stability} is that it is not the standard notion of stability. In \cite{BW13} it is shown that if $\rho(F)<1$ then the dynamical network $(F,X)$ is not only stable but remains stable even if time-delays are introduced into the network's interactions (see
\cite{BW13}). Since the addition of time-delays can have a destabilizing effect on a network, the type of stability considered in Theorem \ref{stability} is a stronger version of the standard notion of stability. To distinguish between these two types of stability, the stability described in Theorem \ref{stability} is given the following name (see \cite{BW13} for more details).

\begin{definition}\label{def:intrinsic} \textbf{(Intrinsic Stability)}
For the dynamical network $(F,X)$, if the spectral radius $\rho(F)<1$ then we say that this network is \emph{intrinsically stable}.
\end{definition}

In the following example we consider the specialization of an intrinsically stable dynamical network.

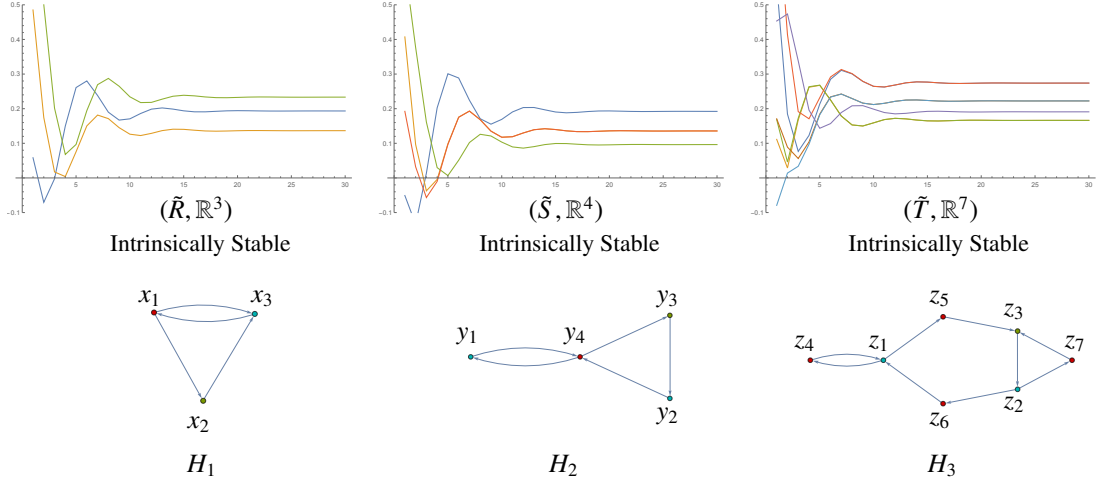
\begin{figure}
\begin{center}
\begin{tabular}{c}
    \begin{overpic}[scale=.25]{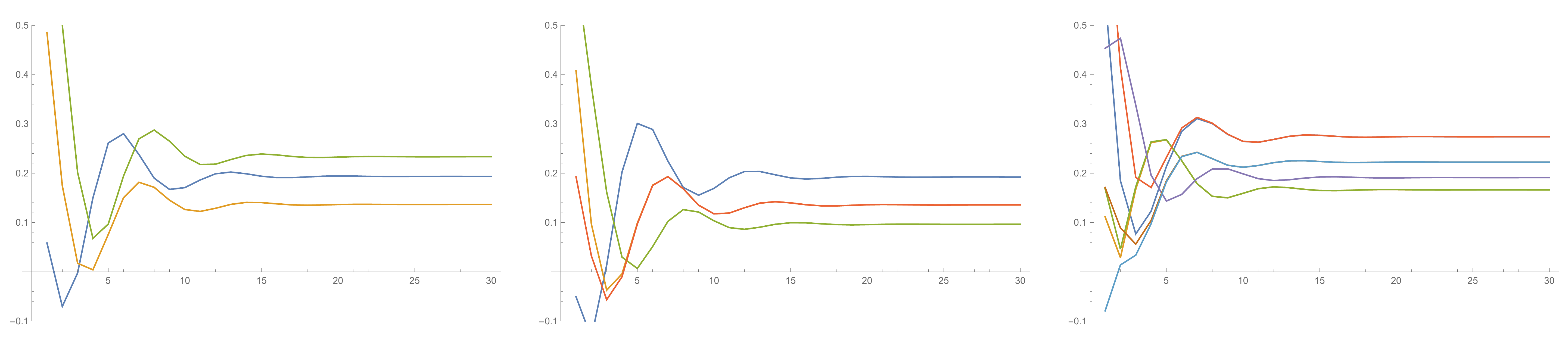}
    \put(10,-2){\small{Intrinsically Stable}}
    \put(14.5,1){$(\tilde{R},\mathbb{R}^3)$}
    \put(43.5,-2){\small{Intrinsically Stable}}
    \put(48,1){$(\tilde{S},\mathbb{R}^4)$}
    \put(77,-2){\small{Intrinsically Stable}}
    \put(82,1){$(\tilde{T},\mathbb{R}^7)$}
    \end{overpic}\\\\
    \begin{overpic}[scale=.18]{NewSpecFig.pdf}
    \put(15.6,-3){$H_1$}
    \put(21.6,13){$x_3$}
    \put(15.6,1.1){$x_2$}
    \put(11,13){$x_1$}

    \put(49.75,-3){$H_2$}
    \put(41.1,9.5){$y_1$}
    \put(60,2){$y_2$}
    \put(60,13){$y_3$}
    \put(51.25,9.5){$y_4$}

    \put(85.5,-3){$H_3$}
    \put(80,8.75){$z_1$}
    \put(92.75,3){$z_2$}
    \put(92.75,11.75){$z_3$}
    \put(73,8.75){$z_4$}
    \put(85.75,13){$z_5$}
    \put(85.75,1.5){$z_6$}
    \put(98.5,8.75){$z_7$}
    \end{overpic}\\\\
\end{tabular}
\caption{Left: The stable dynamics of the network $(\tilde{R},\mathbb{R}^3)$ considered in Example \ref{ex:DRNN} is shown (top) together with its graph of interactions $H_1$ (bottom). Center: The unstable periodic dynamics of the specialization $(\tilde{S},\mathbb{R}^4)$ of $(\tilde{R},\mathbb{R}^3)$ over the set $B=\{1,2\}$ is shown (top) with its graph of interactions $H_2$ (bottom). Right: The unstable but partially synchronizing dynamics of the specialization $(\tilde{T},\mathbb{R}^7)$ of $(\tilde{S},\mathbb{R}^4)$ over the set $C=\{1,2,3\}$ is shown (top) together with its graph of interactions $H_3$. In each of $H_1$, $H_2$, and $H_3$ loops are omitted and vertex colors indicate from which element of the previous network an element was specialized.}\label{Fig:Spec2}
\end{center}
\end{figure}

\begin{example}\label{ex:DRNN}\textbf{(Specialization of an Intrinsically Stable Network)}
Consider the dynamical network $(\tilde{R},\mathbb{R}^3)$, which has the same topology as the network $(R,\mathbb{R}^3)$ in Example \ref{ex:DRNN}, given by
\begin{equation}\label{ex:system}
\tilde{R}\left(\begin{bmatrix}
x_1 \\
x_2 \\
x_3
\end{bmatrix}\right)
=
\left[
\begin{array}{c}
\frac{3}{10}x_1-\frac{1}{2}\tanh(x_3)+\frac{1}{4}\\ [1pt]
\frac{3}{10}x_2+\frac{1}{2}\tanh(x_1)\\ [1pt]
\frac{3}{10}x_3+\frac{1}{2}\tanh(x_1)+\frac{1}{2}\tanh(x_2)
\end{array}\right],
\end{equation}
which is the DRNN in which
\[
f_{ij}(x)=
\begin{cases}
g_i(x)=\tanh(x)  &\text{if} \ i \neq j\\
a_i x +c_i &\text{otherwise}
\end{cases}
\]
where $a_i=\frac{3}{10}$ for $i=1,2,3$; $c_1=\frac{1}{4}$, $c_2=c_3=0$; with the connection weight matrix
\[
\tilde{W}=
\left[\begin{array}{ccc}
0&0&-\frac{1}{2}\\ [1pt]
\frac{1}{2}&0&0\\ [1pt]
\frac{1}{2}&\frac{1}{2}&0
\end{array}\right]
\]
and $A=\tilde{W}+\frac{3}{10}I$. The graph of interactions $H_1$ corresponding to this dynamical network is shown in Figure \ref{Fig:Spec2} (bottom left). The dynamical network $(\tilde{R},\mathbb{R}^3)$ is stable with globally attracting fixed point $\mathbf{x}^*=(.19,.13,.23)$, as can be seen in Figure \ref{Fig:Spec1} (top left). In fact, the network is intrinsically stable with spectral radius $\rho(S)=.962<1$.

If we specialize this network over the base $B=\{1,2\}$, similar to what is done in the previous example, the result is the network $(\tilde{S},\mathbb{R}^4)$ given by
\begin{equation*}
\tilde{S}\left(\begin{bmatrix}
y_1 \\
y_2 \\
y_3 \\
y_4
\end{bmatrix}\right)
=
\left[
\begin{array}{c}
\frac{3}{10}y_1+\frac{1}{2}\tanh(y_4)\\ [1pt]
\frac{3}{10}y_2+\frac{1}{2}\tanh(y_3)\\ [1pt]
\frac{3}{10}y_3+\frac{1}{2}\tanh(y_4)\\ [1pt]
\frac{3}{10}y_4-\frac{1}{2}\tanh(y_1)-\frac{1}{2}\tanh(y_2)+\frac{1}{4}\\
\end{array}\right].
\end{equation*}
Here the network's graph of interaction $H_2=\mathcal{S}_B(H_1)$ is shown in Figure \ref{Fig:Spec1} (bottom center). The vertex colors again indicate from which network element the resulting elements of $H_2$ are specialized. Notably, in contrast to the previous example where specialization destabilized the network's dynamics, the network $(\tilde{S},\mathbb{R}^4)$ is once again intrinsically stable with $\rho(\tilde{S})=.962$ (see Figure \ref{Fig:Spec1}, center bottom).

If the network $(\tilde{S},\mathbb{R}^4)$ is then specialized over the base $C=\{1,3,4\}$ the result in the specialization $(\tilde{T},\mathbb{R}^7)$ given by
\begin{equation*}
\tilde{T}\left(\begin{bmatrix}
z_1 \\
z_2 \\
z_3 \\
z_4 \\
z_5 \\
z_6 \\
z_7
\end{bmatrix}\right)
=
\left[
\begin{array}{c}
\frac{3}{10}z_1+\frac{1}{2}\tanh(z_4)+\frac{1}{2}\tanh(z_6)\\ [1pt]
\frac{3}{10}z_2+\frac{1}{2}\tanh(z_3)\\ [1pt]
\frac{3}{10}z_3+\frac{1}{2}\tanh(z_5)+\frac{1}{2}\tanh(z_7)\\ [1pt]
\frac{3}{10}z_4-\frac{1}{2}\tanh(z_1)+\frac{1}{4}\\ [1pt]
\frac{3}{10}z_5-\frac{1}{2}\tanh(z_1)+\frac{1}{4}\\ [1pt]
\frac{3}{10}z_6-\frac{1}{2}\tanh(z_2)+\frac{1}{4}\\ [1pt]
\frac{3}{10}z_7-\frac{1}{2}\tanh(z_2)+\frac{1}{4}
\end{array}\right].
\end{equation*}
The network's graph of interaction $H_3=\mathcal{S}_C(H_2)$ is shown in Figure \ref{Fig:Spec1} (bottom right). Again one can show that $\rho(\tilde{T})=.962$ so that
\begin{equation}\label{eq:specsame}
\rho(\tilde{R})=\rho(\tilde{S})=\rho(\tilde{T})=.962<1.
\end{equation}
Hence, $(\tilde{T},\mathbb{R}^7)$ is intrinsically stable (see Figure \ref{Fig:Spec1} right bottom). Interestingly, as in the previous example elements $z_4$ and $z_6$ and also elements $z_5$ and $z_7$ synchronize irrespective of the network's initial condition.
\end{example}

Perhaps surprisingly the dynamical network and each of its specializations in the previous example have the same spectral radius. This is not a coincidence but rather a consequence of the similarity between how graph specializations and specializations of dynamical networks are defined. In the proof of the following result we show that if a dynamical network $(F,X)$ is specialized, the resulting dynamical network $(G,Y)$ will have the same spectral radius, i.e. $\rho(F)=\rho(G)$.
For example, each of the networks in Example \ref{ex:DRNN} have the same spectral radius
\[
\rho(R)=\rho(S)=\rho(T)=2.669>1.
\]
Consequently, the following holds.

\begin{theorem}\label{thm:evostability} \textbf{(Stability of Intrinsically Stable Networks Under Specialization)}
Suppose the dynamical network $(F,X)$ given by equation \eqref{eq:netclass} is intrinsically stable. Then for any $B\subset\{1,2,\dots,n\}$ the specialization $(G,Y)$ of $(F,X)$ over $B$ is intrinsically stable.
\end{theorem}

\begin{proof}
Suppose the dynamical network $(F,X)$ with stability matrix $\Lambda$ is intrinsically stable so that, in particular, $\rho(\Lambda)<1$. From Equation \eqref{eq:netclass} it follows that
\[
\Lambda_{ij}=\sup_{\mathbf{x}\in X}\left|\frac{\partial}{\partial x_j}A_{ij}f_{ij}(x_j)\right|= A_{ij}\sup_{x_j\in X_j}\left|\frac{\partial f_{ij}(x_j)}{\partial x_j}\right|
\]
since $A_{ij}\in\{0,1\}$. If $(G,Y)$ is the specialized version of $(F,X)$ over $S$ with stability matrix $\underline{\Lambda}$ then similarly
\[
\underline{\Lambda}_{ij}=\sup_{\mathbf{y}\in Y}\left|\frac{\partial}{\partial y_j}\underline{A}_{ij}f_{\tau(i)\tau(j)}(y_j)\right|=
\underline{A}_{ij}\sup_{y_j\in Y_j}\left|\frac{\partial f_{\tau(i)\tau(j)}(y_j)}{\partial y_j}\right|=
\mathcal{S}_B(A)_{ij}\sup_{y_j\in Y_j}\left|\frac{\partial f_{\tau(i)\tau(j)}(y_j)}{\partial y_j}\right|
\]
since $\underline{A}=\mathcal{S}_B(A)\in\{0,1\}$ where $\tau$ is the origination function of this specialization.

Specializing the stability matrix $\Lambda$ of $(F,X)$ over $B$ results in the matrix
\[
\mathcal{S}_B(\Lambda)_{ij}=\mathcal{S}_B(A)_{ij}\sup_{x_{\tau(j)}\in X_{\tau(j)}}\left|\frac{\partial f_{\tau(i)\tau(j)}(x_{\tau(j)})}{\partial x_{\tau(j)}}\right|=
\mathcal{S}_B(A)_{ij}\sup_{y_j\in Y_j}\left|\frac{\partial f_{\tau(i)\tau(j)}(y_j)}{\partial y_j}\right|
\]
since $y_j=x_{\tau(j)}$, $Y_j=X_{\tau(j)}$. Hence, $\mathcal{S}_B(\Lambda)=\underline{\Lambda}$ and Corollary \ref{cor1} together with Definition \ref{def:matspec} implies that $\rho(\underline{\Lambda})=\rho(\Lambda)<1$ since $\Lambda$ is a nonnegative matrix. Thus, $(G,Y)$ is intrinsically stable.
\end{proof}

Theorem \ref{thm:evostability} is based on the fact that the spectral radius of a graph is preserved under a graph specialization if the graph has nonnegative weights, which every stabilization matrix does (see Corollary \ref{cor1}). The importance of Theorem \ref{thm:evostability} is that it describes, as far as is known to the authors, the first general mechanism for evolving the structure of a network that preserves the network's dynamics. Currently, it is unknown whether any real-world network intrinsically stable. That is, it is unknown whether any naturally occurring network uses intrinsic stability to maintain its function as its structure evolves. However, from a design point of view, if a network is made to be intrinsically stable, rather than just stable, it is guaranteed to to remain stable so long as its topology evolves under the method of specialization described in Section \ref{sec2}. Moreover, if $(F,X)$, given by \eqref{eq:netclass}, is intrinsically stable then any specialization of this network is also intrinsically stable. Hence, any sequence of specializations of the network will result in an intrinsically stable network (cf. Example \ref{ex:DRNN} and \ref{ex:DRNN}).

\begin{corollary}\label{cor:1}
Suppose $(F,X)$ is an intrinsically stable dynamical network given by \eqref{eq:netclass}. If $\{(F^i,X^i)\}_{i=1}^k$ is any sequence of specialization of $(F,X)$ then $(F^k,X^k)$ is intrinsically stable.
\end{corollary}

A natural and open question is whether this is the case for other types of dynamics exhibited by networks, i.e. whether stronger versions of dynamical behavior such as multistability, periodicity, and synchronization, etc. are required for a network to maintain its dynamics as it grows.

\section{Variations of Network Specialization}\label{sec5}

There are a number of variants of the specialization method that preserve both spectral and dynamic properties of a network. One of the most natural is the \emph{partial specialization} of a graph (equivalently dynamical network) in which only a subset of the component branches are specialized. This method is demonstrated in the following example

\begin{figure}
  \begin{overpic}[scale=.275]{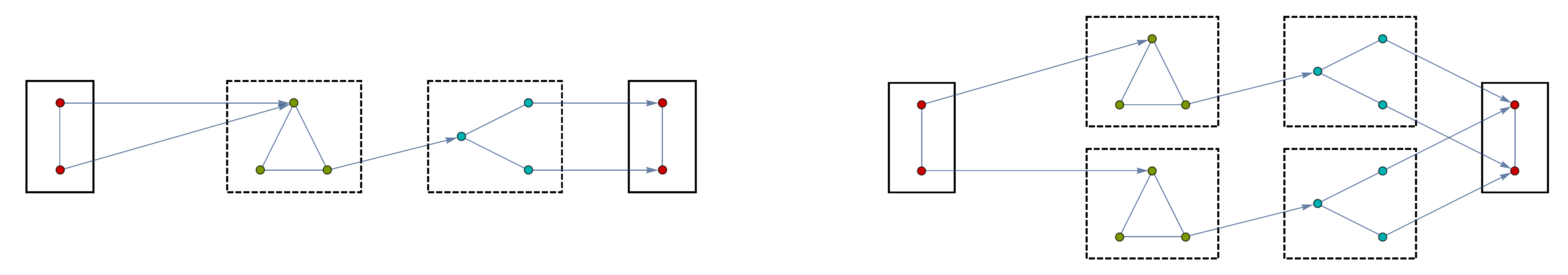}

    \put(15,-0.5){\emph{Unspecialized Graph} $G$}

    \put(3.25,3){$B$}
    \put(18.25,3){$C_1$}
    \put(31,3){$C_2$}
    \put(41.5,3){$B$}


    \put(62,-5){\emph{Partially Specialized Graph} $\mathcal{P}_B(G)$}

    \put(57.5,3){$B$}
    \put(71.75,-1.5){$C_1$}
    \put(84.4,-1.5){$C_2$}
    \put(95,3){$B$}

    \put(79,13){$\rho_1$}
    \put(79,4.5){$\rho_2$}
    \end{overpic}
    \vspace{0.5cm}
  \caption{A partial specialization $\mathcal{P}_B(G)$ of $G$ over the vertex set $B$ is shown. The partial component branches of the specialization are $\rho_1=\beta_1\cup\beta_2$ and $\rho_2=\beta_3\cup\beta_4$, where $\beta_1,\beta_2,\beta_3,\beta_4$ are the component branches in the standard specialization $\mathcal{S}_B(G)$ shown in Figure \ref{Fig:2} (right).}\label{Fig:last0}
\end{figure}
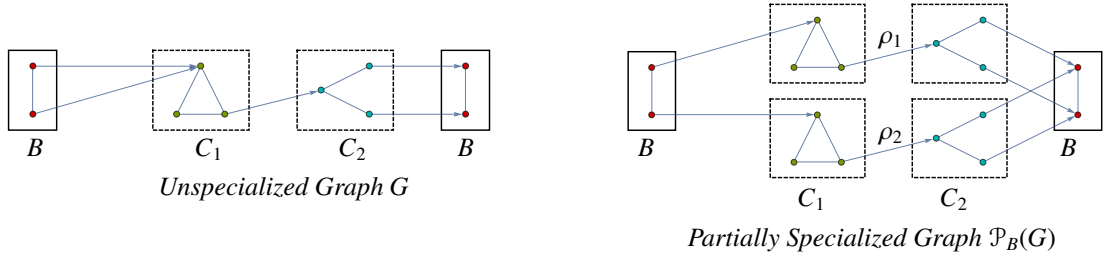

\begin{example}
Let $G$ and $B$ be the graph and base considered in Example \ref{ex:2} and shown in Figure \ref{Fig:last0} (left). In the specialization $\mathcal{S}_B(G)$ there are four component branches $\beta_1,\beta_2,\beta_3,\beta_4\in\mathcal{B}_B(G)$ used to create this graph (see Figure \ref{Fig:2} right). If we instead use the two \emph{partial component branches} $\rho_1=\beta_1\cup\beta_2$ and $\rho_2=\beta_3\cup\beta_4$ the result is the partial specialization $\mathcal{P}_B(G)$ shown in \ref{Fig:last0} (right).
\end{example}

A partial specialization $\mathcal{P}_B(G)$ of a graph $G$ over the base $B$ (equivalently dynamical network) is constructed similar to a standard specialization. First the component branches $\mathcal{B}_B(G)$ are created. Then we choose a partition of $\mathcal{B}_B(G)$, i.e. a collection $\{B_{i=1}^m\}$ such that $\mathcal{B}_B(G)$ is the disjoint union of this set. We then let $\rho_i=\cup_{\beta\in B_i}\beta$ be the $i^{th}$ \emph{partial component branch} for $i=1,\dots,m$. The \emph{partial specialization} $\mathcal{P}_B(G)$ is constructed by \emph{merging}, i.e. identifying, each vertex $v_i\in B$ in any partial branch $\rho_j$ with the same vertex $v_i$ in any other branch $\rho_k$.

We note that the notation $\mathcal{P}_B(G)$ does not uniquely define a specific partial specialization of the graph $G$ as there are typically many ways to choose a partition of the component branches $\mathcal{B}_B(G)$. It is also worth mentioning that the standard specialization method is the partial specialization where the partition we choose is the one in which every subset has a single element. Hence, a graph specialization is the largest partial specialization over a given base. Additionally, an analogous version of Theorems \ref{thm1} and \ref{thm:evostability} as well as Corollary \ref{cor1} hold when we replace the term ``specialization" with ``partial specialization" in these results.

\begin{prop}\label{prop:10}\textbf{(Spectral and Dynamic Consequences of Partial Specialization)}
Let $G=(V,E,\omega)$, $B\subseteq V$, and let $C_1,\dots,C_m$ be the strongly connected components of $G|_{\bar{B}}$. Then for any partial specialization $\mathcal{P}_B(G)$ of $G$ we have
\[
\sigma\big(\mathcal{P}_B(G)\big)=\sigma(G)\cup\sigma(C_1)^{n_1-1}\cup\sigma(C_2)^{n_2-1}\cup\dots\cup \sigma(C_m)^{n_m-1}
\]
where $n_i$ is the number of copies of $C_i$ contained in $\mathcal{P}_B(G)$. Also if $G$ has positive edge weights then
\[
\rho(\mathcal{S}_B(G))=\rho(G).
\]
Additionally, if the dynamical network $(F,X)$ given by equation \eqref{eq:netclass} is intrinsically stable, then for any $B\subset\{1,2,\dots,n\}$ the partial specialization $(G,Y)$ of $(F,X)$ over $B$ is intrinsically stable.
\end{prop}

A proof of Proposition \ref{prop:10} is given in Section \ref{appendix}.

Aside from partial specializations another class of specializations we consider is what we refer to as thinned specializations. A \emph{thinned specialization} $\mathcal{T}_B(G)$ of a graph $G$ over the base $B$ (equivalently dynamical network) is a specialization or partial specialization of $G$ in which we do not include some number of component or partial component branches in the specialized graph.

For example, the specialization of the Wikipedia graph in Figure \ref{Fig:1} is an example of a thinned specialization. The reason is that a number of component branches have been omitted in this specialization. Specifically, the component branches that begin with one color on the left and end with another color on the right are missing. The fact that these branches have been removed in this example makes sense in the context of a disambiguated network since the network segregates by topic. This type of specialization has following spectral and dynamic properties.

\begin{prop}\label{prop:11}\textbf{(Spectral and Dynamic Consequences of Tinned Specialization)}
Let $G=(V,E,\omega)$ have positive edge weights and $B\subseteq V$. Then for any \emph{thinned specialization} $\mathcal{T}_B(G)$ of $G$ we have
\[
\rho(\mathcal{T}_B(G))\leq\rho(G).
\]
Additionally, if the dynamical network $(F,X)$ given by equation \eqref{eq:netclass} is intrinsically stable, then for any $B\subset\{1,2,\dots,n\}$ the thinned specialization $(G,Y)$ of $(F,X)$ over $B$ is intrinsically stable.
\end{prop}

\begin{proof}
Suppose that $G=(V,E,\omega)$ is a graph and $B\subset V$ a base. By Corollary \ref{cor1} it follows that $\rho(\mathcal{S}_B(G))=\rho(G)$. Also, under the assumption that $G$ has positive edge weights, $\mathcal{S}_B(G)$ has positive edge weights. Thus, given that $\mathcal{T}_B(G)$ is a subgraph of $\mathcal{S}_B(G)$ then $\rho(\mathcal{T}_B(G))\leq\rho(\mathcal{S}_B(G))$ (see \cite{HJ90}). Therefore, $\rho(\mathcal{T}_B(G))\leq\rho(G)$. Equivalently, for the matrix $A\in\mathbb{R}^{n\times n}$ the spectral radius $\rho(\mathcal{T}_B(A))\leq\rho(A)$.

Using this fact, if the dynamical network $(F,X)$ is intrinsically stable then $\rho(\Lambda)<1$, where $\Lambda$ is the stability matrix of $(F,X)$. If $(H,Z)$ is the specialization of $(F,X)$ over the base $B\subset\{1,2,\dots,n\}$ then following the proof of Theorem \ref{thm:evostability}, $\underline{\Lambda}=\mathcal{S}_B(\Lambda)$ is the stability matrix of $(H,Z)$. Similarly, one can show that $\underline{\Lambda}_{\mathcal{T}}=\mathcal{T}_B(\Lambda)$ is the stability matrix of the thinned specialization $(G,Y)$. Since
\[
\rho(\underline{\Lambda}_{\mathcal{T}})\leq\rho(\Lambda)<1
\]
then $(G,Y)$ is intrinsically stable.
\end{proof}

As a consequence of Theorem \ref{thm:evostability} and Propositions \ref{prop:10} and \ref{prop:11} any sequence of specializations of a graph, whether these are standard, partial, thinned or a combination of these transformations, results in a graph whose spectral radius is bounded by the spectral radius of the original graph. Similarly, if an intrinsically stable dynamical network is sequentially specialized using any combination of these transformations the result is an intrinsically stable network.

\section{Isospectral Graph Reductions and Proofs}\label{appendix}

In this section we a give a proof of Theorems \ref{thm1} and \ref{prop:0}. The proof of Theorem \ref{prop:0} relies on the theory of isospectral graph reductions, which we give in part here (see \cite{BWBook} for more details). For $M\in\mathbb{R}^{n\times n}$ let $N=\{1,\ldots,n\}$. In what follows, if the sets $T,U\subseteq N$ are proper subsets of $N$, we denote by $M_{TU}$ the $|T| \times |U|$ \emph{submatrix} of $M$ with rows indexed by $T$ and columns by $U$. The isospectral reduction of a square real valued matrix is defined as follows.

A proof of Theorem \ref{thm1} is the following.
\begin{proof}
For the graph $G=(V,E,\omega)$, without loss of generality, let $B=\{v_1,\dots,v_\ell\}$ where $V=\{v_1,\dots,v_n\}$. By a slight abuse in notation we let $B$ denote the index set $B=\{1,\dots,\ell\}$ that indexes the vertices in $B$.

For the moment we assume that the graph $G|_{\bar{B}}$ has the single strongly connected component $S_1$. The weighted adjacency matrix $\mathcal{A}(G)=M\in\mathbb{R}^{n\times n}$ then has the block form
\begin{equation}\label{matrix1}
M=
\left[\begin{array}{cc}
U&Y\\
W&Z
\end{array}\right]
\end{equation}
where $U\in\mathbb{R}^{\ell\times\ell}$ is the matrix $U=M_{BB}$, which is the weighted adjacency matrix of $G|_{B}$. The matrix $Z\in\mathbb{R}^{(n-\ell)\times (n-\ell)}$ is the matrix $Z=M_{\bar{B}\bar{B}}$, which is the weighted adjacency matrix of $G|_{\bar{B}}=S_1$. The matrix $Y=M_{B\bar{B}}\in\mathbb{R}^{\ell\times (n-\ell)}$ is the matrix of edges weights of edges from vertices in $S_1$ to vertices in $G|_B$ and $W=M_{\bar{B}B}\in\mathbb{R}^{(n-\ell)\times \ell}$ is the matrix of edge weights of edges from vertices in $G|_B$ to vertices in $S_1$.

\begin{figure}
\begin{center}
\begin{tabular}{c}
    \begin{overpic}[scale=.43]{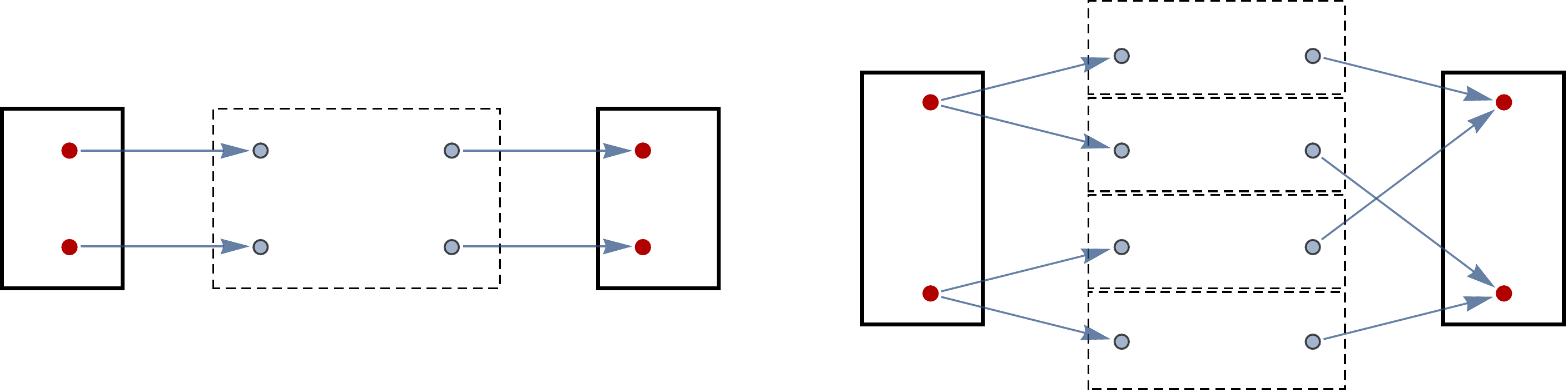}
    \put(21.75,-2){$G$}
    \put(21.5,11.5){$S_1$}
    \put(18,15.25){$v_3$}
    \put(18,8.25){$v_4$}
    \put(25,8.25){$v_6$}
    \put(25,15.25){$v_5$}
    \put(1,15.25){$v_1$}
    \put(1,8.25){$v_2$}
    \put(42,15.25){$v_7$}
    \put(42,8.25){$v_8$}

    \put(74,-3){$\mathcal{S}_B(G)$}
    \put(76,21){$S_1$}
    \put(76,15){$S_1$}
    \put(76,8.5){$S_1$}
    \put(76,2){$S_1$}
    \put(56,18){$v_1$}
    \put(56,6){$v_2$}
    \put(96.5,18){$v_7$}
    \put(96.5,6){$v_8$}

    \put(72,21){$v_3$}
    \put(72,14.75){$v_3$}

    \put(72,8.5){$v_4$}
    \put(72,2.5){$v_4$}

    \put(80.5,21){$v_5$}
    \put(80.5,14.75){$v_6$}

    \put(80.5,8.5){$v_5$}
    \put(80.5,2.5){$v_6$}

    \end{overpic}
\end{tabular}
\end{center}
  \caption{An example of a graph $G=(V,E,\omega)$ with the single strongly connected component $S_1=G|_{\bar{B}}$ is shown (left), where solid boxes together indicate the graph $G|_{B}$. As there are two edges from $G|_{B}$ to $S_1$ and two edges from $S_1$ to $G|_{B}$ there are $2\times 2$ branches in $\mathcal{B}_B(G)$ containing $S_1$. These are merged together with $G|_{B}$ to form $\mathcal{S}_B(G)$ (right).}\label{fig000}
\end{figure}
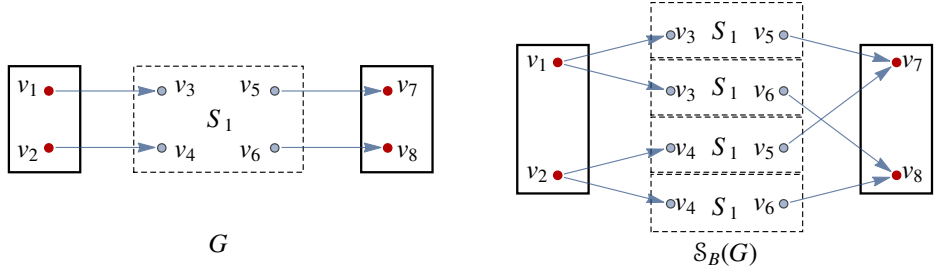

The specialization $\mathcal{S}_B(G)$ is the graph in which all component branches have the form
\[
\beta=v_i,e_{pi},S_1,e_{jq},v_j \ \ \text{for all} \ \ v_i,v_j\in G|_B; \ v_p,v_q\in S_1; \ \ \text{and} \ \ e_{pi},e_{jq}\in E
\]
are merged together with the graph $G|_{B}$ (cf. Figure \ref{fig000}). The weighted adjacency matrix $\hat{M}=\mathcal{A}(\mathcal{S}_B(G))$ has the block form

\begin{equation}\label{matrix2}
\hat{M}=
\left[\begin{array}{cccc}
U&\Big[Y_1 \hspace{0.15in} \cdots\hspace{0.15in} Y_1\Big]&\cdots&\Big[Y_{y} \hspace{0.15in} \cdots \hspace{0.15in} Y_{y}\Big]\\\\

\left[\begin{array}{c}
W_1\\
\vdots\\
W_w
\end{array}\right]&

\left[\begin{array}{ccc}
Z&&\\
&\ddots&\\
&&Z
\end{array}\right]&&0\\
\vdots&&\ddots&\\
\left[\begin{array}{c}
W_1\\
\vdots\\
W_w
\end{array}\right]&0&&

\left[\begin{array}{ccc}
Z&&\\
&\ddots&\\
&&Z
\end{array}\right]\\

\end{array}\right]=
\left[\begin{array}{cc}
U&\hat{Y}\\
\hat{W}&\hat{Z}
\end{array}\right],
\end{equation}
where each $Y_i\in\mathbb{R}^{\ell\times (n-\ell)}$, each $W_j\in\mathbb{R}^{(n-\ell)\times \ell}$, $\sum_{i=1}^y Y_i=Y$ and $\sum_{j=1}^w W_j=W$. Here, $y\geq0$ is the number of directed edges from $S_1$ to $G|_B$. The matrix $Y_i$ has a single nonzero entry corresponding to exactly one edge from this set of edges. Similarly, $w\geq0$ is the number of directed edges from $G|_B$ to $S_1$. The matrix $W_j$ has a single nonzero entry corresponding to exactly one edge from this set of edges. Since there are $w\cdot y$ component branches in $\mathcal{S}_B(G)$ containing $S_1$ then the matrix $\hat{M}\in\mathbb{R}^{(\ell+p(n-\ell))\times (\ell+p(n-\ell))}$ where $p=w\cdot y$.

To prove the result of Theorem \ref{thm1} we will use the Schur complement formula, which states that a matrix with the block form of $M$ has the determinant
\[
\det[M]=\det[Z]\det[U-YZ^{-1}W]
\]
if $Z$ is invertible. Here we apply this formula to the block matrix $M-\lambda I$, which results in the determinant
\begin{align*}
\det[M-\lambda I]&=\det[Z-\lambda I]\det[(U-\lambda I)-Y(Z-\lambda I)^{-1}W]\\
                 &=\det[Z-\lambda I]\det[(U-\lambda I)-\sum_{i=1}^y \sum_{j=1}^w Y_i(Z-\lambda I)^{-1}W_j],
\end{align*}
where the matrix $Z-\lambda I$ is invertible as a result of the proof of Theorem 1.2 in \cite{BWBook}. Similarly, applying the Schur complement formula to the matrix $\hat{M}-\lambda I$ we have the determinant
\begin{align*}
\det[\hat{M}-\lambda I]&=\det[\hat{Z}-\lambda I]\det[(U-\lambda I)-\hat{Y}(\hat{Z}-\lambda I)^{-1}\hat{W}]\\
                 &=\det[Z-\lambda I]^p\det[(U-\lambda I)-\sum_{i=1}^y \sum_{j=1}^w Y_i(Z-\lambda I)^{-1}W_j],
\end{align*}
where the equality
\begin{equation}\label{eq:sums}
\hat{W}(\hat{Z}-\lambda I)^{-1}\hat{Y}=\sum_{i=1}^y \sum_{j=1}^w Y_i(Z-\lambda I)^{-1}W_j
\end{equation}
can be seen by direct multiplication of the block entries in the matrix $\hat{M}$. Hence,
\[
\det[\hat{M}-\lambda I]=\det[Z-\lambda I]^{p-1}\det[M-\lambda I]
\]
implying
\begin{equation}\label{eq:last}
\sigma(\mathcal{S}_B(G))=\sigma(G)\cup\sigma(S_1)^{p-1},
\end{equation}
so that Theorem \ref{thm1} holds if $G|_{\bar{B}}$ has a single strongly connected component $S_1$.

In order to prove Theorem \ref{thm1} in full generality, we will first need to describe an alternative process which generates the specialized graph $\mathcal{S}_B(G)$, which we refer to as ``stepwise specialization.'' We start this process by defining $L_1$ to be the set of the strongly connected components of $G|_{\bar{B}}$, and label the components of this set as $L_1=\{S_{1,1},S_{2,1},\dots,S_{m,1}\}$.  Next we set $G_1=G$ and randomly choose a strongly connected component $S_{i_1,1}\in L_1$ with the property that specializing G around $S_{i_1,1}$ generates a \emph{non-trivial specialization}, i.e. $\mathcal{S}_{\bar{S}_{i,1}}(G_1)\neq G_1$. That is, when $G$, is specialized over $\bar{S}_{i,1}$, the result is a larger graph, meaning that there is more than one directed edge into or out of $S_{i,1}$. We then let $G_2=\mathcal{S}_{\bar{S}_{i,1}}(G_1)$.  Because $S_{i_1,1}$ is a single strongly connected component, we showed in the beginning of this proof that the only new vertices which appear in $G_2$ as $G_1$ is specialized must be $w\cdot y$ copies of $S_{i_1,1}$, where $w$ is the number of edges which are directed into $S_{i_1,1}$ and $y$ is the number of edges which are directed out of $S_{i_1,1}$. Now we relabel the collection of all copies of $S_{i_1,1}$ in the specialized graph $G_2$ as $\{S_{i_1,1}, S_{i_1,2}, \dots, S_{i_1,r_1}\}$.

Next we define $L_2$ to be the collection of all strongly connected components of $G_2|_{\bar{B}}$, which includes all the elements of $L_1$, plus the newly relabeled collection of copies of $S_{i_1,1}$. Thus 
\[
L_2=\{S_{1,1},S_{2,1},\dots,S_{i_1-1,1},S_{i_1,1},S_{i_1,2},\dots,S_{i_1,r_1},S_{i_1+1,1},\dots,S_{m,1}\}.
\]  
We then define $G_3$ by selecting a random element $S_{i_2,j_2}$ of $L_2$ which does not induce a trivial specialization of $G_2$ and set $G_3=\mathcal{S}_{\bar{S}_{i_2,j_2}}(G_2)$.  Again we relabel all of the strongly connected component copies of $S_{i_2,j_2}$ which appear in $G_3$ (this includes new copies and any copies that previously existed in the graph) and define $L_3$ to be the collection of all strongly connected components of $G_3|_{\bar{B}}$.  We continue this process inductively, at each step defining $G_{k+1}=\mathcal{S}_{\bar{S}_{i_k,j_k}}(G_k)$ where $S_{i_k,j_k}$ is randomly chosen from $L_{k}$ and has the property that $\mathcal{S}_{\bar{S}_{i_k,j_k}}(G_k)\neq G_k$. Then we relabel all copies of $S_{i_k,1}$ in $G_k$ as $\{S_{i_k,1},\dots S_{i_k,p_k}\}$ and define $L_{k+1}$ as the strongly connected components of $G_{k+1}$ with all the newly defined labels.

Note that at each step we only specialize over one strongly connected component of $G_{\bar{B}}$. An important observation about this process is that the set of component branches with respect to base $B$ is invariant over each step.  This can be seen by recognizing that there is a one-to-one correspondence between the component branches $\mathcal{B}_B(G_k)$ and $\mathcal{B}_B(G_{k+1})$.  From this fact we can conclude that \begin{equation}\label{amy}
\mathcal{S}_B(G_k)=\mathcal{S}_B(G)
\end{equation}
for all $k\geq 1$,  since we defined specialization to be the graph built by identifying elements of $B$ in all component branches of the graph.

We will continue the process described above until all strongly connected components of $G_K|_{\bar{B}}$ induce a trivial specialization. At this point, the process terminates and we are left with the final graph $G_K$.  We claim that eventually this process must reach this condition and terminate. To verify this, suppose the process did not terminate. Then the specialized graph in each step would be strictly larger than the graph in the previous step.  Thus the number of vertices in $G_k$ must diverge to infinity, which we denote as $\lim_{k\rightarrow\infty}|G_k|=\infty$.  However, we know that by construction $|G_k|\leq|\mathcal{S}_B(G_k)|$. Using Equation \eqref{amy} we know that $|\mathcal{S}_B(G_k)|=|\mathcal{S}_B(G)|<\infty$.  Therefore, at some point  this process must terminate.


When the process does terminate, we claim the final graph $G_K=\mathcal{S}_B(G)$.  To see this, first recall that we defined $\mathcal{S}_B(G)$ to be the graph which consists of the component branches of $\mathcal{B}_B(G)$ in which all elements of $B$ are individually identified (see Definition \ref{def:exp}). Thus, if $S$ is a strongly connected component of $G|_{\bar{B}}$, then each copy of $S$ found in $\mathcal{S}_B(G)$ will have exactly one edge pointing into it and exactly one edge pointing out of it by the definition of a component branch. In the final graph $G_K$, there are no more strongly connected components for which specializing gives a larger graph. Therefore, each strongly connected component in $L_K$ has exactly one edge directing into it and one edge directing out of it. Since this is true for each strongly connected component, then $G_K$ must be composed of a collection of component branches where all corresponding vertices in $B$ are identified.  We already showed $G_K$ has the same component branches with respect to $B$ as $\mathcal{S}_B(G)$.  Using this fact and because strongly connected components of $G_K|_{\bar{B}}$ are each their own component branch, we can conclude that $G_K=\mathcal{S}_B(G)$.

Hence, we can repeatedly apply the result given in Equation \eqref{eq:last} for a single strongly connected component at each step in our ``stepwise specialization process.'' Thus the spectrum of the specialized graph $G_k$ is the spectrum of the graph $G_{k-1}$ from the previous step together with the eigenvalues $\sigma(S_{i_k})^{p_k}$ where $p_k$ is the number of copies of $S_{i_k}$ which were added.  Therefore the eigenvalues of the final graph will be the eigenvalues of the original graph together with the eigenvalues of all the copies of the strongly connected components we added at each step.  That is
$$
\sigma(\mathcal{S}_{B})=\sigma(G)\cup\sigma(S_1)^{n_1-1}\cup\dots\cup\sigma(S_m)^{n_m-1}
$$
where $n_i$ is the number of component branches in $\mathcal{B}_{B}(G)$ containing $S_i$, completing the proof.

\end{proof}

We note that Proposition \ref{prop:10} can be proved analogously to Theorem \ref{prop:0}. The only difference being that we do not fully specialize every branch (cf. Figure \ref{fig000}) but only partially specialize them. The key observation to proving this result is that Equation \ref{eq:sums} also holds in this case in which $Y_i$ and $W_j$ may have more than one nonzero entry but $\sum_{i=1}^y Y_i=Y$ and $\sum_{y=1}^w W_i=W$. Second, the fact that a network with positive edge weights and its partial specialization have the same spectral radius can be shown as in the proof of Corollary \ref{cor1}, which verifies the result.

To prove Theorem \ref{prop:0} we need the notion of an \emph{isospectral graph reduction}, which is a way of reducing the size of a graph while essentially preserving its set of eigenvalues. Since there is a one-to-one relation between the graphs
we consider and the matrices $M\in\mathbb{R}^{n\times n}$, there is also an equivalent theory of \emph{isospectral matrix reductions}. Both types of reductions will be useful to us.

For the sake of simplicity we begin by defining an isospectral matrix reduction. For these reductions we need to consider matrices with rational function entries. The reason is that, by the Fundamental Theorem of Algebra, a matrix
$A\in\mathbb{R}^{n\times n}$ has exactly $n$ eigenvalues including multiplicities. In order to reduce the size of a matrix while at the same time preserving its eigenvalues we need something that carries more information than scalars, which for us are rational functions. The specific reasons for using rational functions can be found in \cite{BWBook}, chapter 1.

We let $\mathbb{W}^{n\times n}$ be the set of $n\times n$ matrices whose entries are rational functions  $p(\lambda)/q(\lambda)\in\mathbb{W}$, where $p(\lambda)$ and $q(\lambda)\neq0$ are polynomials with real coefficients in the
variable $\lambda$. The eigenvalues of the matrix $M=M(\lambda)\in\mathbb{W}^{n\times n}$ are defined to be solutions of the \emph{characteristic equation}
\[
\det(M(\lambda)-\lambda I)=0,
\]
which is an extension of the standard definition of the eigenvalues for a matrix with complex entries.

\begin{definition}\label{def:isored} \textbf{(Isospectral Matrix Reduction)}
The \emph{isospectral reduction} of a matrix $M\in\mathbb{R}^{n\times n}$ over the proper subset $B\subseteq N$ is the matrix
\[
\mathcal{R}_B(M) = M_{BB} - M_{B\bar{B}}(M_{\bar{B}\bar{B}}-\lambda I)^{-1} M_{\bar{B}B}\in\mathbb{W}^{|B|\times|B|}.
\]
\end{definition}

A proof of Theorem \ref{prop:0} part (i) is based on the following result relating the eigenvectors of the a graph $G$ and its reduction $\mathcal{R}_B(G)$.

\begin{theorem}\label{thm:reduction}\textbf{(Eigenvectors of Reduced Matrices)}
Suppose $M\in\mathbb{R}^{n\times n}$ and $B\subseteq N$. If $(\lambda,\mathbf{v})$ is an eigenpair of $M$ and $\lambda\notin \sigma(M_{\bar{B}\bar{B}})$ then $(\lambda,\mathbf{v}_B)$ is an eigenpair of $\mathcal{R}_B(M)$.
\end{theorem}

\begin{proof}
Suppose $(\lambda,\mathbf{v})$ is an eigenpair of $M$ and $\lambda\notin \sigma(M_{\bar{B}\bar{B}})$. Then without loss in generality
we may assume that $\mathbf{v}=(\mathbf{v}_B^T,\mathbf{v}_{\bar{B}}^T)^T$. Since $M\mathbf{v}=\lambda\mathbf{v}$ then
\[
\left[\begin{array}{cc}
M_{BB}&M_{B\bar{B}}\\
M_{\bar{B}B}&M_{\bar{B}\bar{B}}
\end{array}\right]
\left[\begin{array}{c}
\mathbf{v}_B\\
\mathbf{v}_{\bar{B}}
\end{array}\right]=
\lambda \left[\begin{array}{c}
\mathbf{v}_B\\
\mathbf{v}_{\bar{B}}
\end{array}\right],
\]
which yields two equations the second of which implies that
\[
M_{\bar{B}B}\mathbf{v}_B+M_{\bar{B}\bar{B}}\mathbf{v}_B=\lambda\mathbf{v}_{\bar{B}}.
\]
Solving for $\mathbf{v}_{\bar{B}}$ in this equation yields
\begin{equation}\label{eq:bar}
\mathbf{v}_{\bar{B}}=-(M_{\bar{B}\bar{B}}-\lambda I)^{-1}M_{\bar{B}B}\mathbf{v}_{B},
\end{equation}
where $M_{\bar{B}\bar{B}}-\lambda I$ is invertible given that $\lambda\notin \sigma(M_{\bar{B}\bar{B}})$.

Note that
\begin{align*}
(M-\lambda I)\mathbf{v}&=
\left[\begin{array}{c}
(M-\lambda I)_{ BB}\mathbf{v}_{B}+(M-\lambda I)_{ B\bar{B}}\mathbf{v}_{\bar{B}}\\
(M-\lambda I)_{\bar{B}B}\mathbf{v}_{B}+(M-\lambda I)_{\bar{B}\bar{B}}\mathbf{v}_{\bar{B}}
\end{array}\right]\\
&=
\left[\begin{array}{c}
M_{BB}\mathbf{v}_{B}-M_{ B\bar{B}}(M_{\bar{B}\bar{B}}-\lambda I)^{-1}M_{\bar{B}B}\mathbf{v}_{B}\\
M_{\bar{B}B}\mathbf{v}_{B}-(M_{\bar{B}\bar{B}}-\lambda I)(M_{\bar{B}\bar{B}}-\lambda I)^{-1}M_{\bar{B}B}\mathbf{v}_{B}
\end{array}\right]\\
&=\left[\begin{array}{c}
(\mathcal{R}_B(M)-\lambda I)\mathbf{v}_{B}\\
0
\end{array}\right].
\end{align*}
Since $(M-\lambda I)\mathbf{v}=0$ it follows that $(\lambda,\mathbf{v}_B)$ is an eigenpair of $\mathcal{R}_B(M)$.

Moreover, we observe that if $(\lambda,\mathbf{v}_B)$ is an eigenpair of $\mathcal{R}_B(M)$ then by reversing this argument, $\big(\lambda,(\mathbf{v}_B^T,\mathbf{v}_{\bar{B}}^T)^T\big)$ is an eigenpair of $M$ where
$\mathbf{v}_{\bar{B}}$ is given by \eqref{eq:bar}.
\end{proof}

We now give a proof of Theorem \ref{prop:0} part (i).

\begin{proof}
As in the proof of part (i) let $M=\mathcal{A}(G)$ and $\hat{M}=\mathcal{A}(\mathcal{S}_B(G))$ where $B$ is a base of $G$. By slightly abusing our notation we let $B$ also denote the subset of integers $B\subset\{1,\dots,N\}$ that
indexes the vertices in this base. If $(\lambda,\mathbf{v})$ is an eigenpair of $M$ and $\lambda\notin \sigma(M_{\bar{B}\bar{B}})$ then Theorem \ref{thm:reduction} implies that $(\lambda,\mathbf{v}_B)$ is an eigenpair of $\mathcal{R}_B(M)$. The claim is that $(\lambda,\mathbf{v}_B)$ is also an eigenpair of of the reduced matrix $\mathcal{R}_B(\hat{M})$.

To see this, first we note that the matrix $M$ and $\hat{M}$ have the form given in \eqref{matrix1} and \eqref{matrix2}, respectively.
By Definition \ref{def:isored} the reduced matrices $\mathcal{R}_B(M)$ and $\mathcal{R}_B(\hat{M})$ are
\begin{align*}
\mathcal{R}_B(M)&=U-Y[Z-\lambda I]^{-1}W\in\mathbb{W}^{|B|\times|B|}\\
\mathcal{R}_B(\hat{M})&=U-\hat{Y}[\hat{Z}-\lambda I]^{-1}\hat{W}\in\mathbb{W}^{|B|\times|B|}
\end{align*}
Now using equation \eqref{eq:sums} in the proof of Theorem \ref{thm1}, we have
\[
U-\hat{Y}(\hat{Z}-\lambda I)^{-1}\hat{W}=u-\sum_{i=1}^y\sum_{j=1}^w Y_i(Z-\lambda I)^{-1}W_j=Y(Z-\lambda I)^{-1}W
\]
implying $\mathcal{R}_B(\mathcal{S}_B(M))=\mathcal{R}_B(M).$  Using this fact and the observation in the last line of the proof of Theorem \ref{thm:reduction} it follows that $(\lambda,\hat{\mathbf{v}})$ is an eigenpair of $\hat{M}$ where
\[
\hat{\mathbf{v}}=
\left[
\begin{array}{c}
\hat{\mathbf{v}}_B\\
\hat{\mathbf{v}}_{\bar{B}}
\end{array}
\right]=
\left[
\begin{array}{c}
\mathbf{v}_B\\
-(\hat{M}_{\bar{B}\bar{B}}-\lambda I)^{-1}\hat{M}_{\bar{B}B}\mathbf{v}_{B}
\end{array}
\right].
\]
Note that $\mathbf{v}_B=\hat{\mathbf{v}}_B$, which completes the proof.
\end{proof}

We now prove parts (ii) and (iii) of Theorem \ref{prop:0}. We begin with a proof of part (ii)

\begin{proof}
Let $Z$ be a strongly connected component of $G|_{\overline{B}}$ and let $Z_i$ be a copy of $Z$ in $\mathcal{S}_B(G)$. Then there exists a component branch $\beta$ corresponding to $Z_i$, of the form
$\beta = \{v_j, e_0, C_1, ..., C_k, e_k, Z, e_{k+1} C_{k+1}, ..., C_{n}\}$ . Then, because specialization isolates each component branch, we can write the adjacency matrix $A$ of $\mathcal{S}_B(G)$ in the form,

\[
A =
\begin{bmatrix}
\mathcal{A}(In(\beta,Z)) & W \\
Y & X
\end{bmatrix}
=
\begin{bmatrix}
\underline{B} & & & & & & W\\
Y_0 & \underline{C}_1 \\
    & Y_1 	& \underline{C}_2 \\
    &     	& \ddots 	& \ddots \\
    &  	  	&		 	&	Y_{k-1} 	& \underline{C}_k\\
  &  &  	  	&		 	&	Y_k 	& \underline{Z}_i\\
  &  &		&			&		& Y_{k+1} & X \\
\end{bmatrix}.
\]
Here, $\underline{B} = \mathcal{A}(G|_B)$, $\underline{Z}_i = \mathcal{A}(Z_i)$ and for $1 \leq j \leq k$, $\underline{C}_j = \mathcal{A}(C_j)$. For $0 \leq j \leq k+1$, the matrix $Y_j$ is a single entry matrix corresponding to the edge $e_j \in \beta$.
The matrix $X$ is the adjacency matrix for $G$ restricted to nodes that are not in $In(\beta,Z)$, $W$ contains information about edges from nodes represented in $X$ to nodes in $B$.

Suppose $(\lambda, \mathbf{v})$ is an eigenpair of $G$, i.e. $A\mathbf{v}=\lambda\mathbf{v}$, such that $\lambda\notin\sigma(G|_{\bar{B}})$.
As $ \sigma(G|_{\bar{B}}) = \cup_{j=1}^m \sigma(C_i)$ where $C_1,C_2,\dots,C_m$
are the strongly connected components of $G|_{\bar{B}}$ then $\lambda$ is not an eigenvalue of any strongly connected component of $\beta$.

We may conformally partition $\mathbf{v}$ into $\mathbf{v} = [\mathbf{v}_B, \mathbf{v}_{C_1}, \cdots \mathbf{v}_{C_k}, \mathbf{v}_{{Z_{i}}}, \mathbf{v}_{X}]^{T}$ so that entries in each piece correspond with the appropriate sub-matrix of $A$. Then,
applying the eigenvector equation produces,

\[
A\mathbf{v} =
\begin{bmatrix}
\underline{B} & & & & & & W\\
Y_0 & \underline{C}_1 \\
    & Y_1 	& \underline{C}_2 \\
    &     	& \ddots 	& \ddots \\
    &  	  	&		 	&	Y_{k-1} 	& \underline{C}_k\\
  &  &  	  	&		 	&	Y_k 	& \underline{Z}_i\\
  &  &		&			&		& Y_{k+1} & X \\
\end{bmatrix}
\begin{bmatrix}
\mathbf{v}_B \\
{\mathbf{v}_{C}}_1\\
{\mathbf{v}_{C}}_2\\

\vdots \\
\mathbf{v}_{C_k}\\
\mathbf{v}_{{Z_{i}}} \\
\mathbf{v}_{X}
\end{bmatrix}
= \lambda
\begin{bmatrix}
\mathbf{v}_B \\
\mathbf{v}_{C_1}\\
\mathbf{v}_{C_2}\\

\vdots \\
\mathbf{v}_{C_k}\\
\mathbf{v}_{{Z_{i}}} \\
\mathbf{v}_{X}
\end{bmatrix}
=\lambda \mathbf{v}.
\]

We solve for $\mathbf{v}_Z$ in terms of $\mathbf{v}_B$. Multiplying $\mathbf{v}$ by the appropriate block row of $A$ gives $Y_k \mathbf{v}_{C_k} + \underline{Z}_i$, implying that
\[
\mathbf{v}_{Z_{i}} = \lambda \mathbf{v}_{Z_{i}}
\]
\begin{equation*}
\mathbf{v}_{Z_i} = (\lambda I - \underline{Z}_i)^{-1} Y_k \mathbf{v}_{C_k}.
\end{equation*}
In a similar manner for $2\leq i \leq k$, we may solve for $\mathbf{v}_{C_i}$ in terms of $\mathbf{v}_{C_{i-1}}$ producing,
\[
\mathbf{v}_{C_i} = (\lambda I - \underline{C}_i)^{-1} Y_{i-1} \mathbf{v}_{C_{i-1}}.
\]
Combining this with the previous equation yields,
\[
\mathbf{v}_{Z_i} = (\lambda I - \underline{Z}_i)^{-1} Y_k (\lambda I - \underline{C}_k)^{-1} Y_{k-1} \cdots Y_{2}(\lambda I - \underline{C}_2)^{-1}Y_{1}\mathbf{v}_{C_{1}}.
\]
We solve for $\mathbf{v}_{C_{1}}$ in a similar manner and find that
\[
\mathbf{v}_{C_{1}} = (\lambda I - \underline{C}_1)^{-1}Y_0 \mathbf{v}_{B}.
\]
Then,
\[
\mathbf{v}_{Z_i} = (\lambda I - \underline{Z}_i)^{-1} Y_k (\lambda I - \underline{C}_k)^{-1} Y_{k-1} \cdots Y_{1}(\lambda I - \underline{C}_1)^{-1}Y_{0}\mathbf{v}_{B}.
\]
Since $Y_0$,...,$Y_k$, $C_1$ ... $C_k$ and $Z_i$ are all the appropriate submatrixes of $\mathcal{A}(\, In(\beta,Z_i)\,)$, by definition of the eigenvector transfer matrix
\[
\mathbf{v}_{Z_i} = (\lambda I - \underline{Z}_i)^{-1} Y_k (\lambda I - \underline{C}_k)^{-1} Y_{k-1} \cdots Y_{1}(\lambda I - \underline{C}_1)^{-1}Y_{0}\mathbf{v}_{B}
= T(\beta,Z_i,\lambda) \mathbf{v}_{B}.
\]

Note that each inverse in this equation exists since $\lambda\notin\sigma(Z_i)$ and $\lambda\notin\sigma(C_j)$ for $j=1,2,\dots,k$. This completes the proof.
\end{proof}

To prove part (iii) of Theorem \ref{prop:0} we use the following. Given a graph $H = (V,E,\omega)$ that is not strongly connected, let $S = \{S_1, S_2,...,S_k\}$ be the strongly connected components of $H$. If there exist edges $e_1 \cdots e_{k-1}$
such that, $e_j$ is an edge from $S_j$ to $S_{j+1}$ for $1 \leq j \leq k-1$, we call the ordered set $\alpha$ = \{$S_1$, $e_1$, $S_2,...,S_k$\} a partial component branch from $S_1$ to $S_k$ in $H$. We let $\mathcal{P}(S_1,S_k)$ denote the set of all
partial component branches from $S_1$ to $S_k$ in $H$. We let $\mathcal{P}(S_i,S_i)$ be the set containing only $\alpha=\{S_i\}$.

\begin{definition}\textbf{(Partial Eigenvector Transfer Matrix)}
Let $H = (V,E,\omega)$ be a graph that is not strongly connected and let $S$ and $T$ be strongly connected components of $H$. For $\alpha = \{S, e_0, C_1, e_1, ... C_m, e_m  T \} \in \mathcal{P}(S,T)$, let the adjacency matrix of $\alpha$ be

\[
\begin{bmatrix}
\underline{S} \\
Y_0 & \underline{C}_1 \\
    & Y_1 & \underline{C}_2 \\
    &     & \ddots & \ddots \\
    &	  & 		   & Y_m & \underline{T}
\end{bmatrix}
\]
Where $\underline{S} = \mathcal{A}(S)$, $\underline{T} = \mathcal{A}(T)$ and $\underline{C}_i = \mathcal{A}(C_i)$ for $1 \leq i \leq m$. We call the matrix
\[
P(\alpha,\lambda) = (\lambda I - T)^{-1}Y_m (\lambda I - C_m)^{-1}Y_{m-1} \cdots Y_m (\lambda I - T)^{-1}\]
the \emph{partial eigenvector transfer matrix} of $\alpha$, where $\lambda$ is a spectral parameter.
\end{definition}

\begin{lemma}\label{lem:1}
Suppose G = $(V,E,\omega)$ is not strongly connected with strongly connected components $S_1$, $S_2$,...,$S_k$ and
\[
A =
\mathcal{A}(G) =
\begin{bmatrix}
\underline{S}_1 \\
Y_{21} & \underline{S}_2 \\
\vdots & & \ddots \\
Y_{k1} & \hdots & Y_{kk-1} & \underline{S}_k \\
\end{bmatrix}
\] where $\underline{S}_i= \mathcal{A}(S_i)$ for $1 \leq i \leq k$. If $\lambda \notin \sigma(G)$ then,

\[
(\lambda I - A)^{-1} =
\begin{bmatrix}
X_{11} \\
X_{21} & X_{22} \\
\vdots & \vdots & \ddots \\
X_{kl} & X_{k2} & \cdots & X_{kk} \\
\end{bmatrix}
\]
where
\[ X_{ij} = \sum_{\alpha \in \mathcal{P}(S_j,S_i)} P(\alpha,\lambda). \]
\end{lemma}

\begin{proof}
Since $A$ is block lower triangular, $(\rho I -A)$ and $(\rho I - A)^{-1}$ are also block lower triangular. Also, since we can write,
\[
(\rho I - A)
=
\begin{bmatrix}
(\rho I - S_1) \\
-Y_{21} & (\rho I - S_2) \\
\vdots & & \ddots \\
-Y_{k1} & \hdots & -Y_{kk-1} & (\rho I - S_k) \\
\end{bmatrix}
\]
we can write $(\rho I - A)^{-1}$ as,
\[
(\rho I - A)^{-1} =
\begin{bmatrix}
X_{11} \\
X_{21} & X_{22} \\
\vdots & \vdots & \ddots \\
X_{kl} & X_{k2} & \cdots & X_{kk} \\
\end{bmatrix}
\]
where $X_{ii}$ has the same dimensions as $S_i$ for $1 \leq i \leq k$ and $X_{ij}$ has the same dimensions as $Y_{ij}$ when $1 \leq j < i \leq k$. It must then be the case that,

\[(\rho I-A)(\rho I-A)^{-1}=I_1\oplus I_2\oplus \dots,\oplus I_k\]

where each $I_i$ is the identity matrix with the same dimensions as $S_i$. Let $j \in  \{1,2,...,$k$\}$. Then,
\[
\begin{bmatrix}
(\rho I - S_1) \\
-Y_{21} & (\rho I - S_2) \\
\vdots & & \ddots \\
-Y_{k1} & \hdots & -Y_{kk-1} & (\rho I - S_k) \\
\end{bmatrix}
\begin{bmatrix}
0 \\
\vdots \\
0 \\
X_{jj} \\
X_{(j+1)j} \\
\vdots \\
X_{kj} \\
\end{bmatrix}
=
\begin{bmatrix}
0\\
\vdots \\
0 \\
I_j\\
0\\
\vdots\\
0\\
\end{bmatrix}.
\]
Multiplying the $j$th row of $(\rho I -A)$ by the $j$th column of $(\rho I -A)^{-1}$ produces
\[
(\rho I - S_j)X_{jj} = I_j
\]
\begin{equation*}
X_{jj} = (\rho I - S_j)^{-1} = P(\{S_j\},\lambda).
\end{equation*}
Since $\mathcal{P}(S_j,S_j)$ only contains $\{S_j\}$ by definition, it follows that
\[
X_{jj} = \sum_{\alpha \in \mathcal{P}(S_j,S_j)} P(\alpha,\lambda).
\]

We will show by induction that
\[
X_{ij} = \sum_{\alpha \in \mathcal{P}(S_j,S_i)} P(\alpha,\lambda)
\]
when $i > j$. As a base case, consider $X_{(j+1)j}$. By multiplying row $(j+1)$ of $(\rho I-A)$ by the $j$th column of $(\rho I - A)^{-1}$, we obtain the equations
\[
-Y_{(j+1)j}(\rho I - S_{j})^{-1} + (\rho I -S_{(j+1)})X_{(j+1)j} = 0
\]
\begin{equation}\label{eq:new1}
X_{(j+1)j} = (\rho I - S_{(j+1)})^{-1}Y_{(j+1)j}(\rho I - S_j)^{-1}.
\end{equation}
Let $m$ be the number of nonzero entries in $Y_{(j+1)j}$. Then we can write
\[Y_{(j+1)j} = \sum_{k=1}^{m} Y_{(j+1)j}^{(k)}\]
where each $Y_{(j+1)j}^{(k)}$ has one non-zero entry and that entry is equal to a non-zero entry of $Y_{(j+1)j}$. Thus,
\begin{align*}
X_{(j+1)j} &= (\rho I - S_{j+1})^{-1} \sum_{k=1}^{m} Y_{(j+1)j}^{(k)}(\rho I - S_j)^{-1}\\
           &= \sum_{k=1}^{m}(\rho I - S_{j+1})^{-1}  Y_{(j+1)j}^{(k)}(\rho I - S_j)^{-1}.
\end{align*}
For each $k$, the block-diagonal matrix
\[
\begin{bmatrix}
S_j \\
Y_{(j+1)j}^{(k)} & S_{j+1}\\
\end{bmatrix}
\]
is an adjacency matrix for $\alpha^{(k)}=\{S_{j},e^{(k)},S_{j+1}\}$
where $e^{(k)}$ is an edge from $S_j$ to $S_{j+1}$. Then $\alpha^{(k)} \in \mathcal{P}(S_j,S_{j+1})$ and
\[
X_{(j+1)j} = \sum_{k=1}^{m}P(\alpha^{(k)}).
\]
Clearly, $\cup_{k=1}^n\{\alpha^{(k)}\} \subset \mathcal{P}(S_j,S_{j+1})$. We assert that $\cup_{k=1}^n\{\alpha^{(k)}\} = \mathcal{P}(S_j,S_{j+1})$. To verify this, let $\alpha \in \mathcal{P}(S_j,S_{j+1})$. Then $\alpha = \{S_j,e,S_{j+1} \}$ because if
$\alpha$ contained any other strongly connected component $S_l$, it would imply that a path exists from $S_j$ to $S_l$ to $S_{j+1}$ and because of the structure of $A$, it must be the case that $l < j$ or $j+1 < l$. If an edge existed from $S_j$ to $S_{l}$
to $S_{j+1}$ the matrix A would have an entry above the diagonal. This is a contradiction. Thus,
\[
X_{(j+1)j} = \sum_{\alpha \in In(S_j,S_{j+1})}P(\alpha).
\]

By the induction hypothesis assume that when $i < n$ we have
\[
X_{(j+i)j} = \sum_{\alpha \in In(S_j,S_{j+i})}P(\alpha).
\]
Consider $X_{(j+n)j}$. By multiplying the $j+n$th row of $(\rho I - A)$ by the $j$th column of $(\rho I - A)^{-1}$ we obtain the equations
\[
-Y_{(j+n)j}(\rho I - S_j)^{-1} -Y_{(j+n)(j+1)}X_{(j+1)j} \cdots -Y_{(j+n)(j+n-1)}X_{(j+n-1)j} + (\rho I - S_{j+n})X_{(j+n)j} = 0
\]
\begin{equation}\label{eq:new2}
X_{(j+n)j} = (\rho I - S_{j+n})^{-1} Y_{(j+n)j}(\rho I - S_j)^{-1} + \sum_{i=1}^{n-1}(\rho I - S_{j+n})^{-1}Y_{(j+n)(j+i)}X_{(j+i)j}
\end{equation}

As shown in the base case, the first term can be broken up into a sum of partial eigenvector transfer matrices. Let $m_0$ be the number of non-zero entries in $Y_{(j+n)j}.$ If we define $Y_{(j+n)j}^{(k)}$ so that each $Y_{(j+1)j}^{(k)}$ has a single
nonzero entry that is equal to a distinct non-zero entry of $Y_{(j+1)j}$ and
\[Y_{(j+n)j} = \sum_{k=1}^{m_0} Y_{(j+n)j}^{(k)}\]
then,
\[
(\rho I - S_{j+n})^{-1} Y_{(j+n)j}(\rho I - S_j)^{-1} = \sum_{k=1}^{m_0}(\rho I - S_{j+n})^{-1} Y_{(j+n)j}^{(k)}(\rho I - S_j)^{-1}.
\]
For each $1 \leq k \leq m_0 $, $(\rho I - S_{j+n})^{-1} Y_{(j+n)j}^{(k)}(\rho I - S_j)^{-1}$ is the partial eigenvector transfer matrix for a distinct partial component branch $\alpha$ in $\mathcal{P}(S_{j+n},S_j)$ that does not contain any strongly
connected components except $S_{j+n}$ and $S_j$. Let $D_0$ denote the set of all such branches. By definition of an adjacency matrix,  $D_0$ contains one branch for each non zero entry in $Y_{(j+n)j}$. Thus,
\begin{equation}\label{eq:new3}
(\rho I - S_{j+n})^{-1} Y_{(j+n)j}(\rho I - S_j)^{-1} = \sum_{\beta \in D_0}P(\beta,\lambda).
\end{equation}

We consider the other terms in the sum (3). Let $1 \leq i \leq n-1$. By the induction hypothesis,

\[
(\rho I - S_{j+n})^{-1}Y_{(j+n)(j+i)}X_{(j+i)j} = \sum_{\alpha \in \mathcal{P}(S_j,S_{j+i})}(\rho I - S_{j+n})^{-1}Y_{(j+n)(j+i)} P(\alpha,\lambda).
\]
Let $m_i$ represent the number of nonzero entries in $Y_{(j+n)j+i}$. As before we write $Y_{(j+n)j+i}$ as a sum of $m_i$ single entry matrices.
$Y_{(j+n)j+i} = \sum_{k=1}^{m_i} Y_{(j+n)(j+i)}^{(k)}$. Then,
\[
(\rho I - S_{j+n})^{-1}Y_{(j+n)(j+i)}X_{(j+i)j} = \sum_{\alpha \in \mathcal{P}(S_j,S_{j+i})} \sum_{k=1}^{m_i}(\rho I - S_{j+n})^{-1}Y_{(j+n)(j+i)}^{(k)} P(\alpha,\lambda).
\]
It is clear that for each $\alpha \in \mathcal{P}(S_j,S_{j+i})$ and $k$, the term
\[
(\rho I - S_{j+n})^{-1}Y_{(j+n)(j+i)}^{(k)} P(\alpha,\lambda )
\]
is a partial eigenvector transfer matrix for some incoming branch $\gamma \in \mathcal{P}(S_j,S_j+n)$, because the matrix $Y_{(j+n)(j+i)}^{(k)}$ is non zero if and only if an edge exists from $S_{j+i}$ to $S_{j+n}$. If $\mathcal{P}(S_j,S_{j+i})$ is non
empty, there is a branch from $S_j$ to $S_{j+i}$, implying that there must be a branch from $S_j$ to $S_{j+n}$ with partial eigenvector transfer matrix $(\rho I - S_{j+n})^{-1}Y_{(j+n)(j+i)}^{(k)} P(\alpha )$.

What we see here is that for a given $i$, the term \[(\rho I - S_{j+n})^{-1}Y_{(j+n)(j+i)}X_{(j+i)j} \] is equal to the sum of all centrality transfer matrices for the branches in $\mathcal{P}(S_j,S_{j+n})$ that pass through $S_{j+i}$ immediately before
reaching $S_{j+n}$. Let $D_i$ denote the set of all such branches. Then,
\begin{equation*}
(\rho I - S_{j+n})^{-1}Y_{(j+n)(j+i)}X_{(j+i)j} = \sum_{\gamma \in D_i} P(\gamma,\lambda).
\end{equation*}
Putting \eqref{eq:new1}, \eqref{eq:new2}, and \eqref{eq:new3} together gives,
\[
X_{(j+n)j} = \sum_{\beta \in D_0} P(\beta,\lambda) + \sum_{i=1}^{n-1}\sum_{\gamma \in D_i}P(\gamma,\lambda)
\]
Let $D = \cup_{i=0}^{n-1}D_i$. Then
\[
X_{(j+n)j} = \sum_{\alpha \in D}P(\alpha,\lambda)
\]

Clearly, $D \subset \mathcal{P}(S_j,S_{j+n})$. We assert that $D = \mathcal{P}(S_j,S_{j+n})$. To show this, let $\alpha \in \mathcal{P}(S_j,S_{j+n})$. If $\alpha$ has only two components, then $\alpha = \{S_j,e,S_{j+n}\}$ for some edge $e$ and
$\alpha \in D_0 \subset D$ by definition of $D_0$. If $\alpha$ has more then two components, then it has a second to last component, $S_{j+i}$ where $ 1 \leq i \leq n-1 $. By definition of $D_i$, $\alpha \in D_i$. Thus,
\[
X_{j+n,j} = \sum_{\alpha \in \mathcal{P}(S_j,S_{j+n})}P(\alpha,\lambda)
\]
This concludes the proof.
\end{proof}

We now give a proof of part (iii) of Theorem \ref{prop:0}.

\begin{proof}
Since each $Z_k$ is a copy of $Z$ in $\mathcal{S}_B(G)$, each $Z_k$ must correspond to a unique component branch. Let $D = \{\beta_1,...,\beta_{\ell}\}$ be the set of such branches indexed so that $\beta_k$ is the unique branch corresponding to $Z_k$ for
$1 \leq k \leq \ell$. Because each $Z_k$ has the same outgoing branch, it follows that $In(\beta_{j},Z) \neq In(\beta_{k},Z)$ when $j \neq k$. Otherwise, there would exist $k \neq j$ such that $\beta_k = \beta_j$ which contradicts uniqueness of each $\beta_k$.

If $In(\beta_k,Z) = \{v_i,e_0,C_1,e_1,...,e_{n-1}, C_m, e_m, Z \}$ define $F_k = \{C_1, C_2,...C_m \}$
to be the set of all strongly connected components of $G|_{\overline{B}}$ in $In(\beta_k,Z)$.
 We let $F = \cup_{k=1}^{\ell}F_k$. Thus $F$ is the set of all strongly connected components of $G|_B(Z)$ that appear before $Z$ in some incoming branch of $D$.

We may order $F = \{C_1, C_2, ..., C_n,\}$ such that if $1 \leq i < j \leq n$ there are no paths from $C_j$ to $C_i$. If no such ordering existed, it would imply for some $1 \leq i < j \leq n$, paths exist both from $C_i$ to $C_j$ and from $C_j$ back
to $C_i$. This implies that $C_i$ and $C_j$ must be part of the same strongly connected component which is a contradiction.

Thus, we can write the adjacency matrix $A=\mathcal{A}(G)$ of $G$ in the form
\[
A
=
\begin{bmatrix}
\underline{B} & W_{BT} & W_{BZ} & W_{BX} \\
Y_{LB} & L \\
Y_{ZB} & Y_{ZL} & \underline{Z} \\
Y_{XB} & Y_{XL} & Y_{XZ} & X \\
\end{bmatrix}
\]
where $L$ is of the form,
\[
L
=
\begin{bmatrix}
\underline{C}_1 \\
Y_{21} & \underline{C}_2 \\
\vdots & & \ddots \\
Y_{k1} & \hdots & Y_{kk-1} & \underline{C}_n \\
\end{bmatrix}
\]

and $\underline{B} = \mathcal{A}(G|_B)$, $\underline{Z} = \mathcal{A}(Z)$, $\underline{C}_i =\mathcal{A}(C_i)$ for $1\leq i \leq m$. Since, $\lambda$ is an
eigenvalue of $G$, there is a vector $\mathbf{u}$ such that $A\mathbf{u} = \lambda \mathbf{u}$. We may partition $\mathbf{u}$ into $\mathbf{u} = [\mathbf{u}_B, \mathbf{u}_{L}, \mathbf{u}_{Z}, \mathbf{v}_{X}]^{T}$ so that the number of entries in each
sub-vector corresponds with the size of the appropriate sub-matrix of $A$

We apply the eigenvector equation to solve for $\mathbf{u}_{Z}$. Given that
\[
\begin{bmatrix}
\underline{B} & W_{BT} & W_{BZ} & W_{BX} \\
Y_{LB} & L \\
Y_{ZB} & Y_{ZL} & \underline{Z} \\
Y_{XB} & Y_{XL} & Y_{XZ} & X \\
\end{bmatrix}
\begin{bmatrix}
\mathbf{u}_B \\
\mathbf{u}_{L} \\
\mathbf{u}_{Z} \\
\mathbf{u}_{X}
\end{bmatrix}
=
\lambda
\begin{bmatrix}
\mathbf{u}_B \\
\mathbf{u}_{L} \\
\mathbf{u}_{Z} \\
\mathbf{u}_{X}
\end{bmatrix}
\]
we have
\[
Y_{ZB}\mathbf{v}_B + Y_{ZL}\mathbf{u}_L + \underline{Z} \mathbf{u}_Z = \lambda \mathbf{u}_Z
\]
\[
\mathbf{u}_{Z} = (\lambda I - \underline{Z})^{-1}Y_{ZB} \mathbf{u}_{B} + (\lambda	I - \underline{Z})^{-1}Y_{ZL} \mathbf{u}_{L}.
\]
Solving for $\mathbf{u}_L$ produces,
\[
\mathbf{u}_{L} = (\lambda I - L)^{-1}Y_{LB} \mathbf{u}_{B}.
\]
Thus,
\[
\mathbf{u}_{Z} = \big{(}\,(\lambda I - \underline{Z})^{-1}Y_{ZB} + (\lambda	I - \underline{Z})^{-1}Y_{ZL} (\lambda I - L)^{-1}Y_{LB} \,\big{)}\mathbf{u}_B.
\]
We now show that,
\[
\mathbf{u}_{Z} = \big{(}\,(\lambda I - \underline{Z})^{-1}Y_{ZB} + (\lambda	I - \underline{Z})^{-1}Y_{ZL} (\lambda I - L)^{-1}Y_{LB} \,\big{)}\mathbf{u}_B = \sum_{k=1}^{\ell} T(\beta_k,Z,\lambda) \mathbf{u}_B
\]
by showing
\begin{equation}\label{eq:new4}
(\lambda I - \underline{Z})^{-1}Y_{ZB} + (\lambda	I - \underline{Z})^{-1}Y_{ZL} (\lambda I - L)^{-1}Y_{LB} = \sum_{k=1}^{\ell} T(\beta_k,Z,\lambda).
\end{equation}

First, we consider $(\lambda I - \underline{Z})^{-1}Y_{ZB}$. Since every non-zero entry in $Y_{ZB}$ corresponds to an edge from $G|_B$ to $Z$ we let $n_0$ equal the number of non-zero entries in $Y_{ZB}$ and write,
\[
Y_{ZB} = \sum_{r=1}^{n_o}Y_{ZB}^{(r)}
\]

where each $Y_{ZB}^{(r)}$ has exactly one non-zero entry and that entry is equal to a non-zero entry of $Y_{ZB}$. Thus,
\[
(\lambda I - \underline{Z})^{-1}Y_{ZB} = \sum_{r=1}^{n_o}(\lambda I - \underline{Z})^{-1} Y_{ZB}^{(r)}.
\]
We fix $r \in \{1,...,n_0\}$ and consider $(\lambda I - \underline{Z})^{-1} Y_{ZB}^{(r)}$. The matrix $Y_{ZB}^{(r)}$ corresponds to exactly one edge from $G|_B$ to $Z$. Therefore, there must exist $\beta_r \in D$ such that $In(\beta_r,Z) = \{v,e,Z\}$
where $v \in B$ and $e$ corresponds with the non-zero entry in $Y_{ZB}^{(r)}$. Hence,
\[
(\lambda I - \underline{Z})^{-1} Y_{ZB}^{(r)} = T(\beta_r,Z,\lambda).
\]
Furthermore, $\beta_r$ must be the only branch in $D$ that contains $e$. If not, there must be another $\beta_s \in D$ such that $In(\beta_s,Z) =    \{v,e,Z\} = In(\beta_r,Z)$. Since $\beta_r$ and $\beta_s$ are in $D$, they have the same outgoing
branch and it must be the case that $\beta_r = \beta_s$. This contradicts uniqueness of each $\beta$ in $D$.

Thus, each $Y_{ZB}^{(r)}$ corresponds with exactly one $\beta_r \in D$. Furthermore, we see that $In(\beta_r,Z)$ does not contain any strongly connected components except $Z$, because the branch contains an edge directly from a node in $B$
to $Z$. Thus, we may write it's adjacency matrix as follows:
\[
\mathcal{A}(In(\beta_r,Z))
=
\begin{bmatrix}
\underline{B}\\
Y_{ZB}^{(r)} & \underline{Z} \\
\end{bmatrix}.
\]
Then  $(\lambda I - \underline{Z})^{-1} Y_{ZB}^{(r)}$ must the eigenvector transfer matrix of $\beta_r$ with respect to $Z$.

Consider the set $D_0 = \{\beta_1, ... \beta_{n_0} \}$ of component branches corresponding to $\{Y_{ZB}^{(1)}, ..., Y_{ZB}^{(n_0)} \}$. In this case
\begin{equation}\label{eq:new5}
(\lambda I - \underline{Z})^{-1}Y_{ZB} = \sum_{r=1}^{n_o}(\lambda I - \underline{Z})^{-1} Y_{ZB}^{(r)} = \sum_{\beta \in D_0} T(\beta,Z,\lambda).
\end{equation}

Note that for all $\beta \in D_0$, $In(\beta,Z)$ has no strongly connected components except $Z$.
We assert that $D_0$ is the set of all branches in $D$ that satisfy this property. To show this, note that if $\beta \in D$ and $In(\beta,Z) = \{v,e,Z\}$, then
by definition of $\mathcal{A}(G)$, $e$ must correspond to a non-zero entry in $Y_{ZB}$ and therefore $\beta$ corresponds to $Y_{ZB}^{(r)}$ for some
$1 \leq r \leq n_0$. Then $\beta \in D_0$ and $D_0$ is the set of all $\beta \in D$ where $In(\beta,Z)$ contains no strongly connected components except $Z$.

Next we consider $(\lambda	I - \underline{Z})^{-1}Y_{ZL} (\lambda I - L)^{-1}Y_{ZB}$ from equation \eqref{eq:new4}. Once again, we write $Y_{ZL} = \sum_{s=1}^{n_1}Y_{ZL}^{(s)}$ and $Y_{LB}=\sum_{t=1}^{n_2}Y_{LB}^{(t)}$ as the sum of their non-zero entries.

By definition of an adjacency matrix, it must be the case that the set $\{Y_{ZL}^{(s)} \}_{s=1}^{n_1}$ is in a bijective correspondence with the edges from the components in $\{C_1, ... C_n\}$ to $Z$ and the set $\{Y_{LB}^{(t)} \}_{t=1}^{n_2}$
is in a bijective correspondence with the edges from $G|_B$ to components
in $\{C_1, ... C_n\}$. Thus we may let $\{f_s\}_{s=1}^{n_1}$ be the set of edges corresponding with $\{Y_{ZL}^{(s)} \}_{s=1}^{n_1}$ and $\{g_s\}_{t=1}^{n_1}$ be
the set of edges corresponding with $\{Y_{LB}^{(s)} \}_{t=1}^{n_1}$.

We, therefore, have
\[
(\lambda	I - \underline{Z})^{-1}Y_{ZL} (\lambda I - L)^{-1}Y_{LB}
= \sum_{s=1}^{n_1} \sum_{t=1}^{n_2} (\lambda	I - \underline{Z})^{-1}Y_{ZL}^{(s)} (\lambda I - L)^{-1}Y_{LB}^{(t)}.
\]
Fix, $s \in \{1,...,n_1\}$ and $t \in \{1,...n_2\}$ and consider, $Y_{ZL}^{(s)} (\lambda I - L)^{-1}Y_{LB}^{(t)}$. By definition $Y_{LB}^{(s)}$ represents an edge from $G|_B$ to some $C_i \in F$  and $Y_{ZL}^{(s)}$ represents an
edge from some $C_j \in F$ to $Z$. Since there are no paths from $C_j$ to $C_i$ when $i > j$, it must be the case that $i \leq j$. This gives us information about the location of the non-zero entries in $Y_{LB}^{(s)}$ and $Y_{ZL}^{(s)}$. In particular, it must be the case that
\[
\begin{bmatrix}
Y_{LB}^{(t)} & L \\
 &Y_{ZL}^{(s)}
\end{bmatrix}
=
\begin{bmatrix}
\begin{bmatrix}
0 \\
0 \\
Y^{(t)}\\
0 \\
0 \\
0 \\
0 \\
0 \\
\end{bmatrix}
&
\begin{bmatrix}
\underline{C}_1 \\
\vdots & \ddots \\
Y_{i1} & & \underline{C}_i \\
\vdots & & & \ddots \\
Y_{j1} & & & & \underline{C}_j \\
\vdots & & & & & \ddots \\

Y_{k1} & \hdots & Y_{ki} & \hdots & Y_{kj} & \hdots & \underline{C}_n \\
\end{bmatrix}
\\
&
\begin{bmatrix}
\, \, 0 \, \, & \hdots & \hdots & \, \, 0  & \, \, Y^{(s)} & 0 \quad & 0
\end{bmatrix}
\end{bmatrix}. \\
\]
For some $Y^{(t)}$ and $Y^{(s)}$ that contain a single non-zero entry.
This is because $Y_{LB}^{(t)}$ and $Y_{ZL}^{(s)}$ represent an edge from $G|_B$ to a component $C_i$ and an edge from a component $C_j$ to $Z$ respectively. Thus, we conclude that all entries in $Y_{LB}^{(t)}$ and $Y_{ZL}^{(s)}$
that correspond to edges to or from components besides $C_i$ and $C_j$ respectively must be zero.

By lemma \ref{lem:1}
\[
(\lambda I - L)^{-1}
=
\begin{bmatrix}
X_{11} \\
X_{21} & X_{22} \\
\vdots & \vdots & \ddots \\
X_{kl} & X_{k2} & \cdots & X_{kk} \\
\end{bmatrix}
\]
where
\[ X_{ij} = \sum_{\alpha \in \mathcal{P}(C_j,C_i)} P(\alpha,\lambda).
\]

Using this fact we simplify the expression $Y_{ZL}^{(s)}(\lambda I -L)^{-1}  Y_{LB}^{(t)}$ to
\[
\begin{bmatrix}
0 & \hdots &  0  & Y^{(s)} & 0 & \hdots & 0
\end{bmatrix}
\begin{bmatrix}
X_{11}\\
\vdots & \ddots \\
X_{i1} & & X_{ii} \\
\vdots & & \vdots & \ddots \\
X_{j1} & & X_{ji} & & X_{jj} \\
\vdots & & \vdots & & & \ddots \\

X_{k1} & \hdots & X_{ki} & \hdots & X_{kj} & \hdots & X_{nn} \\
\end{bmatrix}
\begin{bmatrix}
0 \\
0 \\
Y^{(t)}\\
0 \\
0 \\
0 \\
0 \\
0 \\
\end{bmatrix}
\]

demonstrating that,
\[
Y_{ZL}^{(s)}(\lambda I -L)^{-1}  Y_{LB}^{(t)}
=
Y^{(s)} X_{ji} Y^{(t)}
=   \sum_{\alpha \in \mathcal{P}(C_i,C_j)} Y^{(s)}P(\alpha,\lambda) Y^{(t)}.
\]

Thus, for each $s$ and $t$ the term,
\[
(\lambda	I - \underline{Z})^{-1}Y_{ZL}^{(s)} (\lambda I - L)^{-1}Y_{LB}^{(t)}
=
\sum_{\alpha \in \mathcal{P}(C_{i_t},C_{j_s})}(\lambda	I - Z)^{-1}Y^{(s)} P(\alpha,\lambda) Y^{(t)}
\]
where $j_s, i_t \in \{1,...,n\}$.

We now show for each $\alpha \in \mathcal{P}(C_{i_t},C_{j_s})$ that
\[
(\lambda	I - \underline{Z})^{-1}Y^{(s)} P(\alpha,\lambda) Y^{(t)}
=
T(\beta,Z, \lambda)
\]
where $\beta \in \mathcal{B}_B(G)$. Let $\alpha \in \ \mathcal{P}(C_{i_t},C_{j_s})$.
Then $\alpha = \{C_{\alpha_0}, e_1, C_{\alpha_1}, e_2, ... C_{\alpha_{N-1}}, e_{N} C_{\alpha_N} \} $
where $C_{\alpha_0} =C_{i_t}$, $C_{\alpha_N} = C_{j_s}$ and for each $1 \leq k \leq N$ we have $C_{\alpha_k} \in F$.
Because $Y^{(s)}$ and $Y^{(t)}$ correspond to unique edges from $C_{j_s}$ to $Z$ and from $G|_B$ to $C_{i_t}$, respectively, there is a unique component branch $\beta \in \mathcal{B}_B(G)$ such that
\[
\beta
=
\{v,g_t,  C_{\alpha_0}, e_1, C_{\alpha_1}, e_2, ... C_{\alpha_{N-1}}, e_N, C_{\alpha_N}, f_s, Z, ... u \}
\]
 for some $v,u \in B$ and where $g_t, f_s \in E$ are edges corresponding to $Y^{(t)}$ and $Y^{(s)}$ respectively as defined previously.
Thus, both
\[
T(\beta,Z, \lambda)
=
(\lambda I -\underline{Z})^{-1} Y^{(s)} (\lambda I -\underline{C}_{\alpha_0})^{-1} Y_1 (\lambda I -\underline{C}_{\alpha_1})^{-1} Y_2 ... Y_N (\lambda I -\underline{C}_{\alpha_N})^{-1} Y^{(t)}
\]
\[
T(\beta,Z, \lambda)
=
(\lambda	I - \underline{Z})^{-1}Y^{(s)} P(\alpha,\lambda) Y^{(t)}
\]
because, $P(\alpha,\lambda) = (\lambda I -\underline{C}_{\alpha_0})^{-1} Y_1 (\lambda I -\underline{C}_{\alpha_1})^{-1} Y_2 ... Y_N (\lambda I -\underline{C}_{\alpha_N})^{-1}$.

We see given $s$, $t$ and $\alpha$ that $(\lambda	I - \underline{Z})^{-1}Y^{(s)} P(\alpha,\lambda) Y^{(t)}$ is equal to the eigenvector transfer matrix for some branch $\beta$ that contains $Z$ where the first edge in $\beta$
is $g_t$ and the edge before $Z$ is $f_s$. Let $D_{st}$ represent the set of all $\beta \in \mathcal{B}_B(G)$ where the first edge is $g_t$ and the edge leading to $Z$ is $f_s$.  We have already shown that for each $\alpha \in
\mathcal{P}(C_{i_t},C_{j_s})$, there exists a $\beta \in D_{st}$ such that
\[
(\lambda	I - \underline{Z})^{-1}Y^{(s)} P(\alpha,\lambda) Y^{(t)}
=
T(\beta,Z, \lambda).
\]

We now show that for each $\beta \in D_{st}$ there exists an $\alpha \in \mathcal{P}(C_{i_t},C_{j_s})$ such that the above equation is true. Let $\beta \in D_{st}$. Then $\beta	 = \{v,g_t, C_{\beta_0}, e_1 ... e_{M},
C_{\beta_M}, f_s, Z ... u\}$ where for $1 \leq k \leq M$, $e_k \in E$ and $C_{\beta_k} \in F$ when $0 \leq k \leq M$. By definition of a component branch, each $e_k$ is and edge from $C_{\beta_{k-1}} $ to $ C_{\beta_k}$ when
$1 \leq k \leq M$. This implies that there is a partial component branch from $C_{\beta_1}$ to $C_{\beta_M}$ of the form, $\alpha = \{ C_{\beta_0}, e_1 ... e_M, C_{\beta_M} \}$. Clearly $\alpha \in \mathcal{P}(C_{\beta_1},C_{\beta_M})$.
Hence, the adjacency matrix for $\alpha$ is of the form,
\[
\mathcal{A}(\alpha)
=
\begin{bmatrix}
\underline{C}_{\beta_0} \\
Y_1 & \underline{C}_{\beta_1} \\
    & \ddots & \ddots \\
    & 		   & Y_m & \underline{C}_{\beta_M}
\end{bmatrix}
\]
and
 $P(\alpha, \lambda) = (\lambda I - \underline{C}_{\beta_M})^{-1} Y_M ... Y_1 (\lambda I - \underline{C}_{\beta_0})^{-1}$. Since the components and edges in $\alpha$ are in $\beta$, the adjacency matrix of $In(\beta,Z)$ can be written as
 \[
 \mathcal{A}(In(\beta,Z))
 =
\begin{bmatrix}
B \\
Y^{(t)} & \underline{C}_{\beta_0} \\
        & Y_1 & \underline{C}_{\beta_1} \\
    		&	& \ddots & \ddots \\
    		&	& 		   & Y_m & \underline{C}_{\beta_M} \\
    		&	&			&	& Y^{(s)} & Z
\end{bmatrix}.
\]
Thus, there exists an $\alpha$ such that $T(\beta,Z,\lambda) = (\lambda I - Z)^{-1} Y^{(s)} P(\alpha, \lambda) Y^{(t)}$. Thus, the sets $\mathcal{P}(C_{i_t},C_{j_s}$ are in a bijective correspondence.

We know that  for each $\alpha \in \mathcal{P}(C_{i_t},C_{j_s})$ there is a $\beta \in D_{st}$ such that
\[
(\lambda	I - \underline{Z})^{-1}Y^{(s)} P(\alpha,\lambda) Y^{(t)}
=
T(\beta,Z,\lambda).
\]

Using this fact, we may substitute each $(\lambda	I - \underline{Z})^{-1}Y^{(s)} P(\alpha,\lambda) Y^{(t)}$ in the sum
\[
\sum_{\alpha \in \mathcal{P}(C_{i_t},C_{j_s})} (\lambda	I - \underline{Z})^{-1}Y^{(s)}P(\alpha,\lambda) Y^{(t)}
\]
for the corresponding $T(\beta,Z,\lambda)$, ($\beta \in D_{st}$ producing
\[
\sum_{\alpha \in \mathcal{P}(C_{i_t},C_{j_s})}(\lambda	I - \underline{Z})^{-1} Y^{(s)}P(\alpha,\lambda) Y^{(t)} = \sum_{\beta \in D_{st}} T(\beta, Z, \lambda).
\]
We know that every $\beta \in D_{st}$ is accounted for in the sum because of the correspondence demonstrated previously.

We have now shown that
\begin{align*}
(\lambda	I - \underline{Z})^{-1}Y_{ZL} (\lambda I - L)^{-1}Y_{LB} &= \sum_{s=1}^{n_1} \sum_{t=1}^{n_2} (\lambda	I - \underline{Z})^{-1}Y_{ZL}^{(s)} (\lambda I - L)^{-1}Y_{LB}^{(t)} \\
&= \sum_{s=1}^{n_1} \sum_{t=1}^{n_2} \sum_{\alpha \in \mathcal{P}(C_{i_t},C_{j_s})}(\lambda	I - \underline{Z})^{-1} Y^{(s)}P(\alpha,\lambda) Y^{(t)}\\
&= \sum_{s=1}^{n_1} \sum_{t=1}^{n_2} \sum_{\beta \in D_{st}} T(\beta, Z, \lambda).
\end{align*}
Since
\[
\mathbf{u}_{Z} = \left(\,(\lambda I - \underline{Z})^{-1}Y_{ZB} + (\lambda	I - \underline{Z})^{-1}Y_{ZL} (\lambda I - L)^{-1}Y_{LB} \,\right)\mathbf{u}_B,
\]
by equation \eqref{eq:new5} we have
\[
\mathbf{u}_{Z}
=
\left(
\sum_{\beta \in D_0} T(\beta,Z,\lambda)
+
\sum_{s=1}^{n_1} \sum_{t=1}^{n_2} \sum_{\beta \in D_{st}} T(\beta, Z, \lambda) \right) \mathbf{u}_B.
\]
Finally we show that
\[
D = D_0 \cup \left( \bigcup_{s=1}^{n_1} \bigcup_{t=1}^{n_2} D_{st} \right).
\]
By definition, $D_0 \subset D$ and $D_{st} \subset D$ for all $1 \leq s \leq n1 ,1 \leq t \leq n_2$. Then $D_0 \cup \left( \bigcup_{s=1}^{n_1} \bigcup_{t=1}^{n_2} D_{st} \right) \subset D$. Let $\beta \in D$. Then $In(\beta,Z)$
either contains a strongly connected component besides Z, or it does not. If it does not, $\beta \in D_0$ as shown previously. If it does, then the first edge in $\beta$ must connect a vertex in $G|_B$ to a vertex in some component
$C_i \in F$ and therefore correspond to $f_t$ for some $1 \leq t \leq n_2$. Additionally the edge that appears prior to $Z$ must correspond to $g_s$ for some $1 \leq s \leq n_1$. Thus $\beta \in D_0 \cup \left( \bigcup_{s=1}^{n_1}
\bigcup_{t=1}^{n_2} D_{st} \right)$ and $D \subset D_0 \cup \left( \bigcup_{s=1}^{n_1} \bigcup_{t=1}^{n_2} D_{st} \right)$. This produces
\[
\mathbf{u}_{Z}
=
\sum_{\beta \in D} T(\beta,Z,\lambda) \mathbf{u}_B
=\sum_{k=1}^{\ell} T(\beta_k,Z,\lambda) \mathbf{u}_B.
\]
By part (ii) of Theorem \ref{prop:0} we have that $T(\beta_k,Z,\lambda) \mathbf{u}_B = \mathbf{v}_{Z_k}$ and
\[
\mathbf{u}_{Z}
=
\sum_{k=1}^{\ell} \mathbf{v}_{Z_k}.
\]
This completes the proof of part (iii).
\end{proof}

\section{Concluding Remarks}\label{conc}

In this paper we consider the interplay of the \emph{dynamics on} a network and the \emph{dynamics of} a network, where the dynamics of the network is modeled by the process of network specialization and the type of dynamics on the network we consider is global stability. What allows us to bridge these two types of dynamics is knowing how specialization effects the spectral properties of the network. In this way we are able to rigourously study the effect the evolving structure of the network has on the dynamics of the network elements.

In a following paper we plan to use this model to study the full interplay of these two types of network dynamics. Specifically, how the dynamics on the network effects the structural evolution of the network and vice-versa and what the eventual structure and dynamics of such networks are.

It is worth mentioning that the extent to which the specialization model models the growth of real-world networks is still unknown. The reason is that, as a network (graph) can be specialized over any subset of its elements (vertices) these are many ways to grow a network especially considering the new variants of network specialization introduced in this paper. An open question is, for specific real-world network, whether a rule can be devised for selecting a base or sequence of bases that grows a small network into a network that has some of the ``finer details" of the real network, e.g. similar features, specific statistics, etc.

Another aspect of the paper that is worth addressing here is the notion of intrinsic stability. An interesting feature of intrinsic stability, not shared by the standard version of stability, is its resilience to change in network structure. Not only is intrinsic stability maintained under specialization but it is also maintained when time delays are introduced into the network's interactions (see \cite{BW13} and \cite{Reber2019}). As time-delays lengthen the paths and cycles between network elements in the network's graph of interactions these time delays have an effect on the network's topology, which are distinct from the topological changes caused by specialization.

An open question is what other types of intrinsic dynamics exist, which analogous to intrinsic stability are maintained under structural changes to the network. Such stronger forms of the standard forms of network dynamics would potentially suggest how real networks maintain their function even as their structure evolves as intrinsic stability does for stable networks.

\section{Acknowledgement} The work of L. A. Bunimovich was partially supported by the NSF grant DMS-1600568. The work of B. Z. Webb is partially supported by was partially supported by the DOD grant HDTRA1-15-0049.

\begin{center}
References
\end{center}

\bibliographystyle{ws-ijbc}
\bibliography{sample1}{}

\end{document}